\documentclass[11pt]{article}
\usepackage{eurosym}
\usepackage[left=2.0cm, right=2.0cm, top=2.2cm, bottom=2.2cm]{geometry}
\usepackage{amsfonts}
\usepackage{amssymb}
\usepackage{amsmath}
\usepackage{lscape}
\usepackage{graphicx}
\usepackage{pgfplots}
\usepackage[round]{natbib}
\usepackage{multirow}
\usepackage[flushleft]{threeparttable}
\usepackage{caption, subcaption}
\usepackage{array}
\usepackage{pdfpages}
\usepackage{rotating}
\usepackage{float}
\usepackage[figurewithin=section]{caption}
\usepackage[tablewithin=section]{caption}
\usepackage{subcaption}
\usepackage{tikz}
\usetikzlibrary{shapes,backgrounds}
\usepackage{verbatim}
\usepackage{url}
\usepackage{setspace}
\usepackage{placeins} 
\usepackage{flafter}  
\usepackage{authblk}
\newtheorem{theorem}{Theorem}[section]

\newtheorem{lemma}[theorem]{Lemma}

\newtheorem{proposition}[theorem]{Proposition}

\newenvironment{proof}[1][Proof]{\noindent \textbf{#1.} }{\  \rule{0.5em}{0.5em}}

\newtheorem{remark}{Remark}[section]

\newtheorem{assumption}{Assumption}

\allowdisplaybreaks
\begin{document}
\title{Decomposing Identification Gains and Evaluating Instrument Identification Power for Partially Identified Average Treatment Effects}
\author[1]{
    Lina Zhang\thanks{Corresponding author. Roetersstraat 11, 1018 WB Amsterdam,
The Netherlands. Email: \texttt{l.zhang5@uva.nl}. This paper was carried out when Zhang was at Monash University. Acknowledgments: Frazier wishes to acknowledge funding from the Australia Research Council (DE200102101), and the Australian center for Excellence in Mathematics and Statistics (ACEMS). Poskitt and Zhao wish to acknowledge funding from the Australia Research Council (DP210103094).}}
\affil[1]{Amsterdam School of Economics, University of Amsterdam, The Netherlands}
\author[2]{David T. Frazier}
\author[2]{D.S. Poskitt}
\author[2]{Xueyan Zhao}\affil[2]{Department of Econometrics and Business Statistics, Monash University, Australia}
\maketitle

\begin{abstract}
This paper examines the identification power of instrumental variables (IVs) for average treatment effect (ATE) in partially identified models. We decompose the ATE identification gains into components of contributions driven by IV relevancy, IV strength, direction and degree of treatment endogeneity, and matching via exogenous covariates. Our decomposition is demonstrated with graphical illustrations, simulation studies and an empirical example of childbearing and women's labour supply. Our analysis offers insights for understanding the complex role of IVs in ATE identification and for selecting IVs in practical policy designs. Simulations also suggest potential uses of our analysis for detecting irrelevant instruments. 
\end{abstract}
\clearpage
\section{Introduction}
The average treatment effect (ATE) is an important policy relevant measure in causal analysis \citep*{heckman2006understanding,imbens2004nonparametric}, but its identification and estimation in empirical research has long been contentious when the treatment is endogenous and instrumental variables (IVs) are used as the identification strategy. This paper takes an empirical causal analyst’s perspective and examines and illustrates the roles of IVs and other factors in the identification and estimation of the ATE within a partially identified modelling framework. Although there have been significant theoretical developments in the econometric literature on the understanding of conventional IV estimands and treatment effect bounds under broader assumptions \citep*[see e.g.][]{manski1990nonparametric,balke1997bounds,heckman1999local,heckman2001instrumental,heckman2005structural,manski2000monotone,chernozhukov2007estimation,chesher2010instrumental}, the exact role of IVs and the associated estimation of various causal effects have remained not well understood in applied economic studies. By synthesising the existing econometric literature on IVs and ATE bounds, and with the help of some novel analyses, in this paper we aim to facilitate a better understanding of the complex role of IVs in ATE identification in practical applications.

In empirical causal analyses, it is common for researchers to estimate the causal effect of an endogenous treatment by a conventional IV estimator as the “identification strategy” \citep[see, e.g.,][]{nunn2014us}. In models with homogeneous treatment responses, using any one of the valid IVs can lead to point identification of the ATE, and conventional IV estimators correctly estimate the ATE. However, as shown by \citet*{heckman2006understanding}, in heterogeneous treatment effect models, different IVs identify different local treatment effects \citep{imbens1994identification}, and conventional IV estimates are no longer robust to the choice of alternative IVs. In fact,  as explained by \citet*{heckman2006understanding}, the classical IV estimand is not ATE but may be a quantity with no easily interpretable meaning regardless of IV strength.

Evidence against homogeneous treatment effects abounds, with estimates based on different sets of valid IVs often producing different treatment effect estimates in practice \citep{carneiro2003understanding,basu2007use,angrist2010extrapolate}.  Once heterogeneous treatment response is allowed, the ATE is often not point identified. Thus, the analysis in \citet*{heckman2006understanding} presents a convincing argument that if the ATE is of primary interest, the identified set for the ATE based on a partially identified model should be preferred to an analysis based on a conventional IV estimation approach. However, there have only been limited applications of partially identified ATE analysis in empirical studies. Consequently, understanding the role of IVs in this setting is important for promoting heterogeneous treatment models for empirical causal analysis.

In heterogeneous treatment effects models, \citet{heckman2001instrumental} demonstrate that the property of ``identification at infinity'' \citep[hereafter IAI,][]{heckman1990varieties}, namely, the availability of IVs that produce propensity scores of zero and one in the limit, leads to point identification of the ATE. However, this condition is rarely satisfied in practice, especially when IVs have limited variation. When IAI fails, inference on the ATE can still be carried out by constructing an identified set for the ATE. Therefore, in a partial identification framework, the impact of IVs can be studied by their influence on the ATE bounds.

We focus on models with binary outcome and binary endogenous treatment. Such  models have been widely used in empirical studies since the pioneering work of \citet{heckman1978dummy}. See  \citet{neal1997effects}, \citet{thornton2008demand} and \citet{ashraf2014household} for examples that use fully parametric bivariate probit models, and see \citet{aakvik2005estimating}, \citet{bhattacharya2008treatment,bhattacharya2012treatment} and \citet{kreider2012identifying,kreider2016identifying} for examples that rely on nonparametric models. The role played by IVs has been a topic of discussion in the literatures, including the notion of ``\emph{identification by functional form}'' \citep*[see e.g.][]{maddala1986limited, wilde2000identification, freedman2010endogeneity, mourifie2014note, han2017identification, li2019bivariate}.\footnote{See \citet{li2019bivariate} for a summary on the topic, including a sufficient condition regarding the support of exogenous regressors for models such as the bivariate probit to achieve ATE point identification without any IVs. However, once the restrictive parametric assumptions fail to hold, IVs become necessary.}

The important role of IVs has been noted for partially identified ATE in heterogeneous treatment effect models \citep[see e.g.][]{manski1990nonparametric, heckman2001instrumental,chesher2005nonparametric,chesher2010instrumental, shaikh2011partial,li2019bivariate}. \citet{heckman2001instrumental} show that for their model with threshold crossing for the treatment, it is the width between the minimum and maximum propensity scores reached by the available instruments that determines the ATE bound width. \citet{chesher2010instrumental} also points out that the support and the strength of the IVs are important in determining the ATE bounds, whilst \citet{li2018bounds} present some simulation results on bound width and IV strength. However, the mechanism through which the IV strength translates to identification gains in partially identified models and whether/how other factors also play a part have not been laid bare in a manner that can be readily understood by practitioners.

In this paper, we examine the role of IVs, as well as their interplay with the degree of endogeneity and exogenous covariates, in the identification of the ATE. Following the partial identification literature,\footnote{For example, see Kitagawa (2009) and Swanson et al (2018) among others.} we use the ATE sign identification and the reduction in the size of the ATE identified set as a measure for \emph{identification gains}. Focusing on the bivariate joint threshold crossing model and the ATE bounds proposed by \citet{shaikh2011partial} (henceforth referred to as the SV model and SV bounds), we disentangle the various factors determining the ATE identification, and provide useful insights for the practitioners into the different sources and natures of identification gains.

To this end, our first contribution is to highlight and demonstrate for the case of SV bounds how IVs achieve identification gains via their attained minimum and maximum conditional propensity scores. The implication to empirical researchers is that, unlike in homogeneous treatment effect models, in heterogeneous models, omitting relevant IVs, or misclassifying continuous IVs as binary ones, could result in a loss of identification power and wider ATE bounds.

Second, we show that, unlike the case of \citet{heckman2001instrumental}, the SV bounds for a binary outcome are additionally impacted by the sign and degree of treatment endogeneity. Interestingly, we find that the endogeneity drives the SV bounds asymmetrically. Specifically, the same propensity score extremes could offer much greater identification power when the ATE and the endogeneity direction are of the opposite signs, relative to the case when they are of the same sign. Thus, it is the interactions of IVs with other features of the model that determine the level of the IV identification power for the ATE. Similar asymmetric influence of treatment endogeneity is also noted in nonlinear parametric models \citep[see, e.g.,][]{freedman2010endogeneity,frazier}.

Our third contribution is to propose a novel decomposition of identification gains into components driven by: (i) the existence of valid IVs that identifies the sign of the ATE; (ii) IV strength that determines the size of the outer set of the ATE identified set; and (iii) the variation of the exogenous covariates that further refines the outer set. The last component is the key driver for achieving the ATE sharp identified set in \citet{mourifie2015sharp} and the ATE point identification in \citet{vytlacil2007dummy}. Based on the decomposition, we further propose a measure for IV identification power (hereafter $IIP$), which captures the critical fact that the IV identification information pertaining to the ATE varies with the endogeneity degree.

The decomposition and $IIP$ analysis allow us to shine a light on the internal workings of the ATE partial identification mechanism and thereby characterize the structure of identification gains. This analysis allows us to offer useful insights on the extent of identification gains achievable by each individual factor, and the contribution of each IV when multiple IVs are used. Our analysis includes graphical illustrations of bound reduction anatomy, as well as Monte Carlo results of finite sample performance. We also apply our methods to an empirical study of women’s labour force participation \citep{angrist1998children}. Two IVs are used in the study: a dummy for the first two children being same-sex siblings (``\emph{Samesex}''), and a dummy for the second birth being a twin (``\emph{Twins}''). Our analysis shows that the identification power of \emph{Twins} is about 1.4 times the identification power of \emph{Samesex} if they are used separately. If both IVs are used, the total IV identification power is only marginally larger than that if only using the \emph{Twins}.

Together with the theoretical decomposition analysis, we believe this paper offers useful insights for empirical causal researchers who wish to understand the complex impacts of IVs on ATE partial identification. Furthermore, we offer some practical examples where our analysis can be used for policy relevant instrument design and selection. Our paper also sheds light on instrument relevancy. Our $IIP$ measure is related to existing approaches in the generalized methods of moment (GMM) literature that seek to determine instrument ``relevancy''. The ability of our approach to rank sets of IVs by their identification gains, in conjunction with our Monte Carlo simulation results lead us to document, we believe for the first time, an important feature of bivariate triangular models: while in the population, adding irrelevant IVs cannot tighten the ATE bounds, in finite-samples, using such IVs could lead to a loss in IV identification power and wider bounds, when the variation of the covariates is small. We liken this phenomena to the well-known problem of irrelevant moment conditions in GMM \citep[see][among others]{breusch1999redundancy,hall2003consistent,hall2005generalized,hall2007information} and leave a more rigorous study of this topic for future research.

The rest of this paper is organized as follows. In Section \ref{sectionFactor} we present the SV model setup and the SV bounds. In Section \ref{sectionBounds} we establish three key factors that affect the ATE bounds. Section \ref{subsectionIG} introduces our decomposition of identification gains and the index of $IIP$. A comprehensive numerical analysis and graphical presentation are given in Section \ref{sectionNA}. Finite sample evaluation and implications for empirical causal practice are presented in Section \ref{sectionER}, and an empirical example is given in Section \ref{emp}. The paper closes in Section \ref{con} with some summary remarks. All proofs are relegated to Appendix.

\section{SV Model Setup and the ATE Bounds}\label{sectionFactor}
Suppose we observe a binary outcome $Y$ and a binary treatment $D$. Consider the joint threshold crossing (JTC) model studied in \citet{shaikh2011partial}:
\begin{equation}\begin{aligned}\label{modelSV}
Y&=1[\nu_1(D,X)>\varepsilon_1],\\
D&=1[\nu_2(X,Z)>\varepsilon_2],
\end{aligned}
\end{equation}
where $X$ denotes a vector of exogenous covariates, $Z$ represents a vector of instruments that can be discrete, continuous or mixed, $\nu_1$ and $\nu_2$ are unknown functions, and $\varepsilon_1$ and $\varepsilon_2$ are unobservable error terms. Let $Y_d=1[\nu_1(d,X)>\varepsilon_1]$ denote the potential outcome for $D=d$ with $d=0,1$. The JTC model allows for flexible forms of heterogeneous treatment effects due to the nonseparable error structure \citep{heckman2006understanding}, and is often used in treatment evaluation studies \citep[see, e.g.,][]{bhattacharya2008treatment,bhattacharya2012treatment,kreider2012identifying}.\footnote{\citet{vytlacil2002independence} shows that the threshold crossing condition is equivalent to the monotonicity assumption. In the JTC models, it means that all individuals with the same observable characteristics will respond to the treatment and instrument in the same direction. \citet{bhattacharya2012treatment} demonstrate that the ATE SV bounds under the JTC model \eqref{modelSV} still hold under a rank similarity condition, a weaker property that allows heterogeneity in the sign of the ATE$(x)$.}

We are interested in the most commonly studied treatment effect, the conditional ATE, defined as
$$\text{ATE}(x)=\mathbb{E}[Y_1|X=x]-\mathbb{E}[Y_0|X=x].$$
For notational simplicity, for any generic random variables $A$ and $B$, henceforth we will use $\text{Pr}[A|b]$ to represent $\text{Pr}[A|B=b]$, unless otherwise stated. The support of $A$ is denoted as $\Omega_A$ and the support of $A$ conditional on $B=b$ is given by $\Omega_{A|b}$.

\begin{assumption}\label{SVR}\textbf{\citep{shaikh2011partial}}
\begin{itemize}
  \item[(a)]$(X,Z)$ is independent of $(\varepsilon_1,\varepsilon_2)$.
  \item[(b)]$(\varepsilon_1,\varepsilon_2)'$ has a strictly positive density with respect to the Lebesgue measure on $\mathbb{R}^2$.
  \item[(c)]The support of the distribution of $(X,Z)$, $\Omega_{X,Z}$, is compact.
  \item[(d)]$\nu_1:\Omega_{D,X}\rightarrow\mathbb{R}$, $\nu_2:\Omega_{X,Z}\rightarrow\mathbb{R}$ are continuous in both arguments.
\item[(e)]The distribution of $\nu_2(X,Z)|X$ is non-degenerate.\label{exclusion}
\end{itemize}
\end{assumption}
Assumption \ref{SVR} (a) and (e) ensure that the instrument $Z$ is independent of the error terms and relevant to the treatment. Conditions (b), (c) and (d) are imposed for analytical simplicity. Denote $P=P(X,Z)=\text{Pr}[D=1|X,Z]$ with support $\Omega_P$. In model \eqref{modelSV}, $Z$ affects the outcome $Y$ only through the propensity score $P$, which is called index sufficiency.

Proposition \ref{prop_SV} summarizes the key results in \citet{shaikh2011partial} and is presented for reference. Denote $\underline{p}:=\inf\{p\in\Omega_{P}\}$ and $\overline{p}:=\sup\{p\in\Omega_{P}\}$.
\begin{proposition}\label{prop_SV}\citep{shaikh2011partial}
\begin{itemize}
  \item[(i)]Under Assumption \ref{SVR} (a) and (b), ATE$(x)\in[L^{SV}(x),U^{SV}(x)]$, where
\begin{equation}
\begin{aligned}
L^{SV}(x)&=\sup_{p\in\Omega_{P|x}}\left\{\text{Pr}[Y=1,D=1|x,p]+\sup_{x'\in \textbf{X}_{1+}(x)}\text{Pr}[Y=1,D=0|x',p]\right\}\\
\label{p40}
&~~~~-\inf_{p\in\Omega_{P|x}}\left\{\text{Pr}[Y=1,D=0|x,p]+p\inf_{x'\in\textbf{X}_{0+}(x)}\text{Pr}[Y=1|x',p,D=1]\right\},
\end{aligned}
\end{equation}
\begin{equation}
\begin{aligned}
U^{SV}(x)&=\inf_{p\in\Omega_{P|x}}\left\{\text{Pr}[Y=1,D=1|x,p]+(1-p)\inf_{x'\in\textbf{X}_{1-}(x)}\text{Pr}[Y=1|x',p,D=0]\right\}\\
\label{p41}
&~~~~-\sup_{p\in\Omega_{P|x}}\left\{\text{Pr}[Y=1,D=0|x,p]+\sup_{x'\in\textbf{X}_{0-}(x)}\text{Pr}[Y=1,D=1|x',p]\right\},
\end{aligned}
\end{equation}where $L^{SV}(x)$ and $U^{SV}(x)$ are referred to as SV lower and upper bound, respectively, and $\textbf{X}_{0+}(x)$, $\textbf{X}_{0-}(x)$, $\textbf{X}_{1+}(x)$, and $\textbf{X}_{1-}(x)$ are subsets of $\Omega_{X}$ whose definition can be found in Appendix \ref{app_def}.\footnote{It is understood that the supremum and infimum operators in the bounds are only taken over regions where all conditional probabilities are well defined: $\text{Pr}[Y=y,D=d|x',p]$ and $\text{Pr}[Y=y|x',p,D=d]$ are well defined, if there exists $z'\in\Omega_{Z|x'}$ such that $\text{Pr}[D=1|x',z']=p$. The supremum over an empty set is defined as 0, and the infimum over an empty set is defined as 1. The same definition of the supremum and infimum also applies to the bounds throughout the paper.}
  \item[(ii)]Under Assumption \ref{SVR} (a) and (b), if $\Omega_{X,P}=\Omega_X\times\Omega_P$ holds, the SV bounds can be simplified to
\begin{equation*}
\begin{aligned}
L^{SV}(x)&=\text{Pr}[Y=1,D=1|x,\overline{p}]+\sup_{x'\in \textbf{X}_{1+}(x)}\text{Pr}[Y=1,D=0|x',\overline{p}]\\
&~~~~-\text{Pr}[Y=1,D=0|x,\underline{p}]-\underline{p}\inf_{x'\in\textbf{X}_{0+}(x)}\text{Pr}[Y=1|x',\underline{p},D=1],
\end{aligned}
\end{equation*}
\begin{equation*}
\begin{aligned}
U^{SV}(x)&=\text{Pr}[Y=1,D=1|x,\overline{p}]+(1-\overline{p})\inf_{x'\in\textbf{X}_{1-}(x)}\text{Pr}[Y=1|x',\overline{p},D=0]\\
&~~~~-\text{Pr}[Y=1,D=0|x,\underline{p}]-\sup_{x'\in\textbf{X}_{0-}(x)}\text{Pr}[Y=1,D=1|x',\underline{p}],
\end{aligned}
\end{equation*}
  \item[(iii)]Under Assumption \ref{SVR} (a) to (d), if $\Omega_{X,P}=\Omega_X\times\Omega_P$ holds, then the SV bounds above are sharp.
  \item[(iv)]Under Assumption \ref{SVR} (a), (b) and (e), the sign of the $\text{ATE}(x)$ is identified by the sign of
$\text{Pr}[Y=1|x,p]-\text{Pr}[Y=1|x,p']$, for any $p,p'\in\Omega_P$ such that $p>p'$,
where sgn$[\cdot]$ is the conventional signum function.
\end{itemize}
\end{proposition}There are several implications of Proposition \ref{prop_SV}. First, the SV bounds in \eqref{p40} and \eqref{p41} consist of two layers of intersection evaluations. The first layer is to intersect all possible values of the conditional propensity score $P$ given $X=x$, or equivalently, of the IVs. The second layer of intersections are taken over values of covariates whose variation can compensate (to some degree) for the variation in the treatment variable in the outcome equation. Thus, both the IVs and the covariates contribute to the ATE partial identification using the SV bounds.

In particular, the SV bounds are informative in the sense that the sign of the ATE$(x)$ is recovered, if $Z$ has nonzero prediction power for the treatment, meaning that there exist at least two different values of $P(x,z),P(x,z')\in\Omega_{P|x}$. It indicates that the IVs' contribution to the SV bounds is achieved in two steps: first, using any nonzero IV variation to identify the ATE sign; and second, using the first layer intersections to further shrink the SV bounds. In addition, when the instrument is binary, the ATE$(x)$ sign is identified by the sign of the ``intention-to-treat'' parameter $\text{Pr}[Y=1|x,P(x,z)]-\text{Pr}[Y=1|x,P(x,z')]$.\footnote{\citet{machado2013instrumental} study how the intention-to-treat parameter can be used for identification and inference of the ATE sign under different assumptions, including the JTC model. \citet{ura2018heterogeneous} and \citet{tommasi2020bounding} also find the intention-to-treat parameter identifies the ATE sign in heterogeneous treatment effect models.} For the contribution of covariates, \citet{vytlacil2007dummy} show that it is possible to achieve the ATE point identification via the SV bounds if $X$ contains a continuous element or the exclusion restriction holds in both equations. \citet{shaikh2011partial} Remark 2.1 also discussed conditions for point identification.

Third, by imposing the support condition $\Omega_{X,P}=\Omega_X\times\Omega_P$, the SV bounds are \emph{sharp} and can be simplified to be functions of the two extreme values of the propensity score $P$.\footnote{The condition $\Omega_{X,P}=\Omega_X\times\Omega_P$ is saying that for any $x,x'\in\Omega_X$, there exist possible realizations $z,z'$ of $Z$ such that $\text{Pr}[D=1|x,z]=\text{Pr}[D=1|x',z']$. The same support condition is also used in \citet{mourifie2014partially}. If $X$ and $Z$ are both binary, then the point identification assumption in \citet{vytlacil2007dummy} is equivalent to $\Omega_{X,P}=\Omega_X\times\Omega_P$.} Using the expressions of the sharp SV bounds, we can confirm that the IAI condition guarantees the ATE point identification in the JTC model: if $\underline{p}=0,\overline{p}=1$ holds, then $L^{SV}(x)=U^{SV}(x)=$Pr$(Y=1|x,\overline{p}=1)-$Pr$(Y=1|x,\underline{p}=0)$. Note that the condition $\Omega_{X,P}=\Omega_X\times\Omega_P$ might fail to hold in practice, especially when the variation in $Z$ is limited. When it does not hold, \citet{mourifie2015sharp} provides an ATE sharp identified set by further exploiting the information on other individuals with different covariate and propensity score values.\footnote{\citet{han2020sharp} provide a linear programming method to compute the sharp bounds with binary scalar instrument. While, this paper considers more general settings where the instrument(s) can be discrete, continuous or mixed.} 

Given the SV bounds, the width of the SV bounds can be defined as
$$\omega^{SV}(x)=U^{SV}(x)-L^{SV}(x)\,.$$
Throughout the paper, we use the sign of the ATE and the reduction in the width of the identification set to measure identification gains and power.

\section{The Determinants of ATE Bounds}\label{sectionBounds}
We first study how three key factors drive the SV bounds: the conditional propensity score, the direction and degree of endogeneity, and the variation of covariates.

\subsection{The Conditional Propensity Score}
As discussed in the introduction, the propensity score is the vehicle that carries the identification information in the IVs. We start from the conditional propensity score (CPS), $P(x,Z)=$Pr$[D=1|X=x,Z]$, and examine its features that determine the ATE bounds. Define the two extremes of the CPS as $$\underline{p}(x):=\inf_{z\in\Omega_{Z|x}}\{p\in\Omega_{P|x,z}\}\;\text{ and }\;\overline{p}(x):=\sup_{z\in\Omega_{Z|x}}\{p\in\Omega_{P|x,z}\}.$$

\subsubsection{With Support Condition}
In this subsection, we assume the support condition $\Omega_{X,P}=\Omega_X\times\Omega_P$ holds. \citet{shaikh2011partial} show that the SV bounds are sharp and are functions of $\underline{p}$ and $\overline{p}$. In addition, the support of the CPS $P(x,Z)$ reduces to $\Omega_P$ for $\forall x\in\Omega_X$. Therefore, the two extreme values of the CPS become to $\underline{p}(x)=\underline{p}$ and $\overline{p}(x)=\overline{p}$, where $\underline{p}=\inf\{p\in\Omega_{P}\}$ and $\overline{p}=\sup\{p\in\Omega_{P}\}$. The proposition below shows that the extreme values of the CPS, i.e. $\underline{p}$ and $\overline{p}$, are crucial determinants of the lower and upper bounds and the width of the ATE sharp identified set.
\begin{proposition}\label{sv1}Under Assumption \ref{SVR} and $\Omega_{X,P}=\Omega_X\times\Omega_P$, for $\forall x\in\Omega_X$,
\begin{itemize}
  \item[(i)]$L^{SV}(x)$ is weakly increasing as $\underline{p}$ decreases or as $\overline{p}$ increases;
  \item[(ii)]$U^{SV}(x)$ is weakly decreasing  as $\underline{p}$ decreases or as $\overline{p}$ increases;
  \end{itemize}
  and hence
  \begin{itemize}
  \item[(iii)]$\omega^{SV}(x)$ is weakly decreasing as $\underline{p}$ decreases or as $\overline{p}$ increases.
\end{itemize}
\end{proposition}
Complementing the results in \citet{shaikh2011partial}, Proposition \ref{sv1} fills in the blanks by \emph{explicitly} showing that when the IAI fails to hold, the contribution of IVs to the ATE bounds are delivered by the magnitude of $\underline{p}$ and $\overline{p}$ (while holding anything else fixed). In other words, the two extreme values of the CPS define the \emph{IV strength} in the JTC model. More importantly, we find a monotone relationship between $\underline{p}$ and $\overline{p}$ and the ATE bounds. We emphasize that this monotone
result is not found in \citet{heckman1999local,heckman2001instrumental}, \citet{chesher2010instrumental} and \citet{shaikh2011partial}.\footnote{\citet{heckman1999local,heckman2001instrumental} show that in the treatment threshold crossing model, the ATE bound width is linearly related to the two extreme values of CPS, while they do not give the similar result for the lower and upper ATE bounds. \citet{chesher2010instrumental} briefly discusses the importance of IV strength and support in determining the extent of ATE identified set for models with threshold crossing in outcome. \citet{shaikh2011partial} only show that the bounds are functions of $\underline{p}$ and $\overline{p}$ in the JTC model.}

The implications of Proposition \ref{sv1} are significant. For homogeneous treatment response models such as linear regression models, any valid and relevant IVs would lead to the same identification gains because the ATE is point identified (e.g. by Wald estimand or 2SLS), therefore missing IVs or misclassifying continuous IVs into binary ones play no role in identifying the ATE \citep{heckman2006understanding}. However, in heterogeneous treatment effect models, the fact that the CPS affects the ATE bounds via the distance of its extreme values to zero and one demonstrates that missing and misspecified IVs will result in narrower span of CPS, which may lead to nontrivial loss in identification power.

\begin{remark}
The concept of IV strength determined by $\underline{p}$ and $\overline{p}$ is different from the conventional IV strength measures in studies of weak IV tests (e.g., measured by the first-stage $F$-stat for continuous endogenous regressors, or the pseudo-$R^2$ for binary response variables). Because the key ingredients of the latter are the correlation between the IVs and the endogenous regressors, and the variation of the IVs to that of the random noise. Proposition \ref{sv1} indicates that two IV sets with the same $\underline{p}$ and $\overline{p}$ will make identical contributions to identification gains in SV sharp bounds (for given direction and degree of endogeneity and covariates), irrespective of their correlation with the endogenous regressors or their variability.
\end{remark}

\subsubsection{Without Support Condition}
The support condition $\Omega_{X,P}=\Omega_X\times\Omega_P$ may fail when the instruments have limited variation. \citet{mourifie2015sharp} establishes the ATE sharp identified set via exploiting the variation of covariates without using the support condition. We are able to establish and verify the relationship between the CPS and a tractable characterization of ATE bounds, which is an outer set of \citet{mourifie2015sharp}'s sharp identified set and also an outer set of the SV bounds.\footnote{The analysis using the outer set is practically useful because empirical studies often construct confidence region based on an outer set \citep[see, e.g.,][]{chesher2020econometric}. Since an outer set always includes the sharp identified set, the analysis is conservative yet valid as long as the model is correctly specified \citep{molinari2020microeconometrics,kedagni2020discordant}. Moreover, discussions in Section \ref{sectionCOV} imply that further improvement in the outer set towards the sharp identified set of \citet{mourifie2015sharp} can be attributed to the information of covariates.} The explicit expressions of the outer set, denoted by $[\underline{L}^{SV}(x),\overline{U}^{SV}(x)]$ can be found in \eqref{svw_p} and \eqref{svw_n}; see the proof of Proposition \ref{svw}. The width of the outer set, denoted by $\overline{\omega}(x)=\overline{U}^{SV}(x)-\underline{L}^{SV}(x)$, is given in the proposition below.

\begin{proposition}\label{svw}Under Assumption \ref{SVR}, for any given $x\in\Omega_X$,
\begin{align*}
&\text{if ATE}(x)>0\,,~\mbox{then}~
\overline{\omega}(x)=\text{Pr}\left[Y=1,D=1|x,\underline{p}(x)\right]+\text{Pr}\left[Y=0,D=0|x,\overline{p}(x)\right]\,;\\
&\text{if ATE}(x)<0\,,~\mbox{then}~
\overline{\omega}(x)=\text{Pr}\left[Y=1,D=0|x,\overline{p}(x)\right]+\text{Pr}\left[Y=0,D=1|x,\underline{p}(x)\right].
\end{align*}
Moreover,
\begin{itemize}
  \item[(i)]$\underline{L}^{SV}(x)$ is weakly decreasing as $\underline{p}(x)$ decreases or as $\overline{p}(x)$ increases;
  \item[(ii)]$\overline{U}^{SV}(x)$ is weakly increasing as $\underline{p}(x)$ decreases or as $\overline{p}(x)$ increases;
  \item[(iii)]$\overline{\omega}(x)$ is weakly decreasing as $\underline{p}(x)$ decreases or as $\overline{p}(x)$ increases.
\end{itemize}
\end{proposition}
We can see that the width of the outer set is monotone in the extreme values of CPS, i.e. $\underline{p}(x)$ and $\overline{p}(x)$. Thus, it is no doubt that the extreme values of CPS are also crucial determinants of the ATE sharp bounds, even though the relation may not be monotone. In addition, Proposition \ref{svw} also confirms that the IAI leads to $\overline{\omega}(x)=0$ and ATE point identification.

\begin{remark}
  Although the magnitude of the CPS extreme values determines the location and width of the ATE bounds, simply comparing the difference $\overline{p}(x)-\underline{p}(x)$ of various IV sets may not be a proper practice to compare instrument strength and predict their resulting ATE bounds. For different IV sets, if their CPS supports are nesting each other, then, ceteris paribus, a larger $\overline{p}(x)-\underline{p}(x)$ produces tighter ATE bounds. However, if their CPS supports do not overlap or partially overlap, even if the difference, say $\overline{p}(x)-\underline{p}(x)$ and $\overline{p}'(x)-\underline{p}'(x)$, are the same, their associated ATE bounds are not directly comparable simply based on this and without further investigation, because it is not clear to what extent the change of the ATE bounds caused by moving $\overline{p}(x)$ to $\overline{p}'(x)$ can be offset by moving $\underline{p}(x)$ to $\underline{p}'(x)$.
\end{remark}

\subsection{The Direction and Degree of Endogeneity}
Next, we illustrate that, for given IV strength (i.e., $\overline{p}(x)$ and $\underline{p}(x)$), how the treatment endogeneity can enhance or hinder the IV identification power and SV bounds.
We introduce a family of bivariate single parameter copulae that specifies the joint distribution of $(\varepsilon_1,\varepsilon_2)$, while we do not require the copula nor the marginal distributions to be known. Denote a copula as $C(\cdot,\cdot;\rho):(0,1)^2\mapsto(0,1)$, where $\rho\in\Omega_\rho$ is a scalar dependence parameter that fully describes the joint dependence between $\varepsilon_1$ and $\varepsilon_2$, and their dependence increases as $\rho$ increases.\footnote{In the special case of a normal bivariate probit model $\rho\in[-1,1]$ represents the correlation between the error terms.} For any given copula, the sign and magnitude of $\rho$ can be understood as the direction and degree of endogeneity.

We also impose additional dependence structure, the concordance ordering, on the copula $C(\cdot,\cdot;\rho)$. Following \citet{joe1997multivariate}, 
for $\rho_1\neq\rho_2$ and $u_1,u_2\in(0,1)^2$, we say that the copula $C(\cdot,\cdot;\rho)$ satisfies the \emph{concordant ordering} with respect to $\rho$, denoted as $C(u_1,u_2;\rho_1)\prec_cC(u_1,u_2;\rho_2)$, if
\begin{align}\label{CO}C(u_1,u_2;\rho_1)\leq C(u_1,u_2;\rho_2),~\text{for any }\rho_1<\rho_2.
\end{align}
The concordant ordering is a stochastic dominance restriction and is embodied in many well-known copulae, including the normal copula. 
Similar stochastic dominance conditions are employed in, e.g., \citet{han2017identification} and \citet{han2019estimation}, to derive identification and estimation results for the parametric bivariate probit model and its generalizations. Denote the cumulative distribution function of any random variable $A$ as $F_{A}$.
\begin{assumption}\label{ass2_sv}The joint distribution of $(\varepsilon_1,\varepsilon_2)'$ is given by a member of the single parameter copula family
$F_{\varepsilon_1,\varepsilon_2}(e_1,e_2)=C(F_{\varepsilon_1}(e_1),F_{\varepsilon_2}(e_2);\rho),$ for $(e_1,e_2)\in\mathbb{R}^2$,
where $C(\cdot,\cdot;\rho)$ satisfies the concordant ordering with respect to $\rho$. \end{assumption}
Assumption \ref{ass2_sv} does not require $C(\cdot,\cdot;\rho)$ nor $F_{\varepsilon_1}$ and $F_{\varepsilon_2}$ to be known, and is used to establish the impacts of the dependence parameter $\rho$ on the SV bounds.\footnote{The same general conclusions are expected to follow in cases where the parametric copula is relaxed to nonparametric copulae with similar ordering conditions.}
\begin{proposition}\label{sv2}Under Assumptions \ref{SVR} and \ref{ass2_sv}, we have that
\begin{itemize}
  \item[(i)]if ATE$(x)>0$, $\overline{\omega}(x)$ is weakly increasing in $\rho$;
  \item[(ii)]if ATE$(x)<0$, $\overline{\omega}(x)$ is weakly decreasing in $\rho$.
\end{itemize}
\end{proposition}
Proposition \ref{sv2} implies that the outer set of the SV bounds is impacted by the direction and the degree of endogeneity. Intuitively, this is because the outer set is constructed using the joint probabilities of the outcome and the treatment. In particular, the effect of endogeneity's direction is asymmetric: given a positive ATE, negative dependence between the two error terms helps narrow down the ATE bound width, while the opposite holds for a negative ATE. Thus, even for IVs that have the same $\underline{p}(x)$ and $\overline{p}(x)$, the information contained in the IVs can be correspondingly scaled via the leverage induced by the direction and degree of endogeneity.

\begin{remark}Proposition \ref{sv2} emphasizes that the \textit{IV strength} is a different concept from the \textit{IV identification power} in JTC models with binary dependent variables. A set of ``seemingly weak'' IVs judged from the extremes of CPS alone, may actually achieve significant identification gains in an empirical context where the ATE and the direction of endogeneity have opposite sign. Thus, the IVs actually have strong identification \textit{power}. Conversely, a set of ``seemingly strong'' IVs can be surprisingly \textit{powerless} when the ATE and the direction of endogeneity have the same sign, resulting in wide ATE bounds. Thus, having knowledge of the signs of the ATE and the direction of endogeneity can help applied researchers to form expectations on the IV identification power and the ATE bound width.
\end{remark}
\subsection{Covariate Support and Variability}\label{sectionCOV}
In this section, we summarize some key results of covariates in the ATE identification literature. As we have seen from the construction of the SV bounds, covariates contribute to partially identifying the ATE. It is therefore a useful insight for empirical researchers that having a rich set of covariates can help shrink the ATE bounds, for given IV strength. Here, the impact of covariates can be considered as an extension of the popular propensity score matching method in conventional empirical studies. In fact, even without the IAI condition, in certain situations, the covariate variability can ensure the point identification of the ATE. In other cases, where the covariate variability is minimal, no further tightening can be achieved beyond the outer set. Define \begin{align*}
\mathcal{X}^0=&\{x\in\Omega_X:~\exists (\tilde{x},\tilde{z},z)~s.t.~\nu_1(1,\tilde{x})=\nu_1(0,x)\text{ and }P(x,z)=P(\tilde{x},\tilde{z})\in\Omega_P\},\\ \mathcal{X}^1=&\{x\in\Omega_X:~\exists (\tilde{x},\tilde{z},z)~s.t.~\nu_1(1,x)=\nu_1(0,\tilde{x})\text{ and }P(x,z)=P(\tilde{x},\tilde{z})\in\Omega_P\}.
\end{align*}
\begin{proposition}\label{prop3_1}\citep[][Remark 2.2]{vytlacil2007dummy,shaikh2011partial}
Under Assumption \ref{SVR},
for any $x\in \mathcal{X}^0\cap \mathcal{X}^1$, we have $\omega^{SV}(x)=0$ and ATE$(x)$ is point identified.
\end{proposition}
Proposition \ref{prop3_1} is a special case of Theorem 4.1 of \citet{vytlacil2007dummy} when the outcome is a binary variable, and it is also discussed by \citet{shaikh2011partial} Remark 2.2. It states that when the variation in $X$ exactly compensates for the variation in $D$, which is satisfied by those values $x\in \mathcal{X}^0\cap \mathcal{X}^1$, the point identification of the ATE can be obtained. Intuitively, the identification is achieved by finding a shift in covariates that offsets a shift in the treatment to make the outcome and the probability of getting treated unchanged. \citet{mourifie2015sharp} also uses similar idea to tighten the SV bounds.
\begin{proposition}\label{prop3_2}
Under Assumption \ref{SVR},
\begin{itemize}
  \item[(i)]\citep{chiburis2010semiparametric} for any given $x\in\Omega_X$, if there exists no $(\tilde{x},\tilde{z},z)$ and $\tilde{x}\neq x$ such that $P(x,z)=P(\tilde{x},\tilde{z})\in\Omega_P$, then $\omega^{SV}(x)=\overline{\omega}(x)$;
  \item[(ii)]if the random variable $\nu_1(D,X)|D$ is degenerate, then $\omega^{SV}(x)=\overline{\omega}(x)$ for all $x\in\Omega_X$.
\end{itemize}
\end{proposition}
Proposition \ref{prop3_2} provides two conditions of covariate, under which the SV bounds fail to utilize any information in $X$ to compensate the impact of the variation in $D$ on the outcome, and therefore the SV bounds coincide with its outer set. Particularly, condition (i) summarizes the results described in Section 3.1.2 in \citet{chiburis2010semiparametric}, when none of other propensity score values match exactly to $P(x,z)$. The condition in (ii) holds, if either $X$ is a void variable, or if $X$ has no impacts on outcome.

The results above and the construction of SV bounds both indicate that for any given IVs strength, if the variation of covariates can be enriched, then SV bounds shrink. Thus, for the given IV strength, any further improvement of the outer set towards SV bounds (or, towards \citet{mourifie2015sharp}'s bounds) can be attributed to additional identification information in covariates, 
whose variations compensate for the treatment's variation to some degree. 

\section{Decomposing Identification Gains and IV Identification Power}\label{subsectionIG}
In this section, we present a novel decomposition of the identification gains of the SV bounds into components attributable to IVs and exogenous covariates. In addition, we introduce the IV identification power index, $IIP$.
\subsection{Decomposing Identification Gains}
We first introduce the benchmark ATE bounds of \citet{manski1990nonparametric} which use no IVs and are often referred to as ``the worst case scenario'' \citep[see][]{tamer2010partial,chiburis2010semiparametric,bhattacharya2012treatment}. The Manski bounds are used as a benchmark, because if IVs are all irrelevant, the ATE SV bounds collapse to the ATE Manski bounds.\footnote{See Remark 2.1 of \citet{shaikh2011partial} and \citet{chiburis2010semiparametric} Corollary 1. In particular, Corollary 1 of \citet{chiburis2010semiparametric} shows that if no IVs are relevant, although the outcome threshold crossing condition tightens the identified set of $\nu_1(0,x)$ and $\nu_1(1,x)$ relative to those under Manski's framework, it fails to improve the ATE bounds.} Denote $L^M(X)$ and $U^M(x)$ as the lower and upper Manski bound of ATE$(x)$, respectively,
\begin{equation}
\begin{aligned}\label{7}
L^M(x)&=-\text{Pr}[Y=1,D=0|x]-\text{Pr}[Y=0,D=1|x],\\
U^M(x)&=\text{Pr}[Y=1,D=1|x]+\text{Pr}[Y=0,D=0|x].
\end{aligned}
\end{equation}
It is apparent that the width of the Manski bounds, defined as $\omega^M(x)=U^M(x)-L^M(x)$, are one, for any $x\in\Omega_X$, with the lower bound and upper bound falling on either side of zero. Our decomposition of identification gains is inspired by the results in Section \ref{sectionBounds}. We start with the Manski bounds, and divide the Manski bounds into four components, which represent incremental identification gains made by the SV bounds over the benchmark Manski bounds, due to different determinants.

 \begin{itemize}
   \item [(i)] $C_1(x)$: \textbf{Contribution of IV Validity}. The first component of the identification gains is due to the identification of the ATE$(x)$ sign. This contribution is accredited to IV validity, since we can identify the sign of the ATE$(x)$ if the IVs are independent of the error term $(\varepsilon_1,\varepsilon_2)$ and $\nu_2(X,Z)|X$ is nondegenerate (or equivalently, if the IVs are valid) regardless of the IV strength.\footnote{If ATE$(x)$=0 is identified, the contribution of SV bounds using any valid IVs already leads to the ATE point identification.} For $\forall x\in\Omega_X$,
       $$C_1(x)=\begin{cases}U^M(x),&\text{ if ATE}(x)<0,\\L^M(x),&\text{ if ATE}(x)>0.\end{cases}$$
   \item [(ii)] $C_2(x)$: \textbf{Contribution of IV Strength}.
   Conditional on the first component, IV validity, the second component captures the further reduction achieved by the SV bounds to its outer set, via intersecting over all possible values of $Z$. This is reflected in the dependence of the SV bounds in \eqref{p40} and \eqref{p41} on the two extreme values of the CPS. The closer the extreme values to 0 and 1 are, the greater is $C_2(x)$. Therefore, identification gains attributed to IV strength can be measured as
       $$C_2(x)=\omega^M(x)-C_1(x)-\overline{\omega}(x).$$
    \item [(iii)] $C_3(x)$: \textbf{Contribution of Covariates}. The third component is the further reduction of the SV bounds beyond $C_1(x)$ and $C_2(x)$, brought about by the variation of exogenous covariates in the outcome equation and its ability to compensate for the effect caused by changes in the treatment, as implied by Proposition \ref{prop3_1} and \ref{prop3_2}. Hence, the component $C_3(x)$ is attributed to the exogenous covariates:
   $$C_3(x)=\overline{\omega}(x)-\omega^{SV}(x).$$
   \item [(iv)]$C_4(x)$: \textbf{Remaining SV Bound Width}. The last component relates to the remaining SV bounds that cannot be further reduced by the observable data under the SV modeling assumptions. This component can be thought of as the signal-to-noise ratio of the error terms. By construction, we have $C_4(x)=\omega^{SV}(x)$.
 \end{itemize}
It is easy to see that $C_1(x)+C_2(x)+C_3(x)+C_4(x)=\omega^M(x)=1$. If $\nu_2(X,Z)|X$ is degenerate and the IVs have no explanatory power for the treatment, then $C_1(x)=C_2(x)=C_3(x)=0$ and the SV bounds reduce to Manski bounds.  In addition, $C_1(x)$ to $C_4(x)$ can always be identified and estimated from the data. In practice, once the model has been estimated (parametrically or non-parametrically), the estimates can be used to construct the decomposition. 
\begin{remark}
If we consider the identification gains of the sharp bounds in \citet{mourifie2015sharp} relative to the benchmark Manski bounds, the decomposition above still holds with changes in expressions of $C_3(x)$ and $C_4(x)$. This is because that the improvement of \citet{mourifie2015sharp} bounds to the SV bounds is due to variation of covariates.
\end{remark}
\begin{remark}
It is worth to note that although we do not decompose the identification gains based on the direction and the degree of endogeneity, the magnitude of all the four components varies with them. According to Proposition \ref{sv2}, the endogeneity affects $\overline{\omega}(x)$, which enters all four components either directly or indirectly due that the summation of the four components is a fixed value one.
\end{remark}
\subsection{IV Identification Power (IIP)}\label{sectionIVIP}
Based on the decomposition, we can then construct a quantitative measurement of IV identification power in the partial identification setting.
For $\forall x\in\Omega_X$, define the IV identification power $IIP(x)$ as
\begin{align}\label{R}
IIP(x):=\begin{cases}\omega^M(x)-\overline{\omega}(x)=C_1(x)+C_2(x),~&\text{if $\nu_2(X,Z)|X=x$ is nondegenerate}\\0,~&\text{if $\nu_2(X,Z)|X=x$ is degenerate}
\end{cases}
\end{align} where $\overline{\omega}(x)$ is the width ATE outer set defined in Proposition \ref{svw} and $IIP(x)\in[0,1]$. Setting $IIP(x)=0$ when $\nu_2(X,Z)|X=x$ is degenerate is equivalent to setting $\overline{\omega}(x)=\omega^M(x)=1$.\footnote{The definition allows $IIP(x)$ to be discontinuous at $\Omega_{P|x}=p_x$ for some constant $p_x\in[0,1]$, i.e. when $\Omega_{P|x}$ is a singleton.}
$IIP(x)$ represents the identification gains that are due to the IVs alone and it can be viewed as an index of the IV identification power. The overall IV identification power can be obtained by $\mathbb{E}_X[IIP(X)]$. The following proposition formalizes some important properties of $IIP(x)$.
\begin{proposition}\label{coro2}
Let Assumption \ref{SVR} (a) to (d) hold.
\begin{enumerate}
             \item[(i)]$IIP(x)=0$ if and only if none of the IVs are relevant; if $IIP(x)=0$ then the SV bounds reduce to the benchmark Manski bounds;
             \item[(ii)]$IIP(x)=1$ if the IVs have perfect predictive power for the treatment $D$ (i.e., IAI), in the sense that there exists $p^*$ and $p^{**}$ in $\Omega_{P|x}$ such that $\text{Pr}[D=1|x,p^*]=0$ and $\text{Pr}[D=1|x,p^{**}]=1$. Moreover, the ATE$(x)$ is point identified when $IIP(x)=1$.
           \end{enumerate}
\end{proposition}
Proposition \ref{coro2} indicates that values of $IIP(x)$ can be compared, across different sets of IVs, or across different values of $x$ given the same set of IVs, since they are standardized relative to the same benchmark.\footnote{$IIP(x)$ or $\mathbb{E}_X[IIP(X)]$ can also be compared across various studies if necessary.} For example, $IIP(x)=0.4$ can be interpreted as that the Manski bounds can be reduced by 40\% by using instruments alone.\footnote{Theoretically, the value of $IIP(x)$ should lie in $[0,1]$ and the width of Manski bounds is always one. Then $IIP(x)$ can be interpreted as the percentage points of the identification gains brought by the IVs. In finite sample settings where the estimated Manski bound width may no longer be exact one, then the sample explanation can be obtained by computing the ratio $\hat{IIP}(x)/\hat\omega^M(x)$ using their estimates.} In addition, the values of $IIP(x)$ at its end points are intuitively interpretable: $IIP(x)=0$ indicates that IVs are completely irrelevant and SV bounds reduce to Manski bounds; and $IIP(x)=1$ when the IVs are able to perfectly predict the treatment status and ATE$(x)$ is point identified.

$IIP(x)$ is a meaningful measure of IV usefulness for improving the ATE partial identification, because it incorporates the impacts of the direction and degree of endogeneity on the ATE bounds. Using $IIP(x)$ instead of the CPS extreme values emphasizes the fact that the latter one cannot provide a full picture of the IV identification power. For example, suppose a set of IVs has a small CPS support and the IVs seem to be ``weak''. However, a large $IIP(x)$ can be achieved if the magnitude of treatment endogeneity is large and it has an opposite sign with that of the ATE. See for example, the numerical examples in Figure \ref{fig:IG0.25} of Section \ref{sectionNA}, when ATE$(x)>0$, $\rho=-0.8$ and $\lambda$ close to zero,\footnote{In the numerical example, $\lambda$ determines the CPS extremes. Closer to zero $\lambda$ means narrower span of the CPS support.} we have $IIP(x)>$50\%. Thus, the identification power of the IV set can actually be very strong. Conversely, if the treatment endogeneity is of the same sign as that of the ATE, a reasonable CPS range may produce barely satisfactory $IIP(x)$ and identification gains. See for example, the numerical examples when ATE$(x)>0$, $\rho=0.8$ and $\lambda$ close to 0.5, we have $IIP(x)<$50\%.

\begin{remark}
Note that $IIP(x)$ ignores the component of identification gains attributable to the exogenous covariates, namely $C_3(x)$. This neglect is reasonable, because the discussion in Section \ref{sectionCOV} demonstrates that if there is no variation in covariates, then $C_3(x)=0$ regardless of the information in IVs. Thus, any further improvement of the SV bounds from its outer set must be via the covariate variation. This indicates that $IIP(x)$ is a measure of identification gains due to IVs alone and it measures the smallest identification power of a given set of IVs.
\end{remark}

\section{Numerical Illustration}\label{sectionNA}
In this section we illustrate numerically and graphically the results on how each determinant affects the SV bounds and the decomposition of identification gains.\footnote{To make the paper manageable, we focus on the SV bounds. Our numerical analysis may be of broader interest as a prototype for evaluation of IV identification power that could be conducted in other frameworks under weaker assumptions, such as those in \citet{chesher2010instrumental}, \citet{heckman2001instrumental}, \citet{manski2000monotone} and \citet{mourifie2015sharp}.} Consider a data generating process (DGP) with a linear additive latent structure, which is similar to that studied in \citet{li2019bivariate}:
\begin{equation}
\begin{aligned}\label{model2}
Y&=1[\alpha D+\beta X+\varepsilon_1>0],\\
D&=1[\gamma Z+\pi X+\varepsilon_2>0],
\end{aligned}
\end{equation}
where $X\perp Z$, $X\sim\mathbb{N}(0,1)$ and $Z\in\{-1,1\}$ with $\text{Pr}(Z=1)=1/2$. In addition, $(X,Z)'\perp(\varepsilon_1,\varepsilon_2)$ where $(\varepsilon_1,\varepsilon_2)$ is zero mean bivariate normal with unit variances and correlation $\rho$. We set $\alpha=1$ and $\pi=0$ across all parameter settings. Given this specification, there is a monotonic one-to-one mapping from the coefficient of the IV, $\gamma$, to the extreme values of the CPS.

\begin{table}[htb!]
\begin{center}
\caption{Parameter Settings in Numerical Illustration}\label{tab:numerical_DGP}
\setlength{\tabcolsep}{4pt}
\vspace{-0.25cm}
		\begin{tabular}{lcl}
			\hline
			\hline
&&Parameter Setting\cr
\hline
extreme values of CPS&&$\gamma=-4\;:\;0.2\;:\;4$\cr
direction and degree of endogeneity&&$\rho=-0.99\;:\;0.05\;:\;0.99$\cr
impact of exogenous covariate&&$\beta\in\{0.05,0.25,0.45\}$\cr
\hline
\hline
\end{tabular}
\end{center}
\vspace{-0.3cm}
\end{table}

We study the SV bounds under different parameter settings displayed in Table \ref{tab:numerical_DGP}. We capture the extreme values of the CPS by $\gamma$, the direction and degree of endogeneity by $\rho$, and the impact of exogenous covariate by $\beta$.
We compute the SV bounds and the Manski bounds, and implement the identification gains decomposition under the true DGP.
In what follows we present the outcomes at $x=\mathbb{E}[X]$.\footnote{Our numerical experiments have been conducted at various quantile points of $X$, but space considerations prevent us from listing all the results. We present the outcomes at $x=E[X]$ as these are representative.}

\subsection{Determination of ATE Bounds}
In Figure \ref{fig:3d}, we plot in the first row the SV bounds for the ATE$(x)$, and in the second row the bound width.
The SV bounds reduce to the Manski bounds when the IVs are irrelevant with $\gamma=0$ (the separate lines in the graphs at $\gamma=0$). When $\gamma$ moves away from zero, the SV bound width has a significant drop. In addition, since ATE$(x)$ is positive ($\alpha>0$), the SV bound width increases as $\rho$ increases. Moreover, comparison of the plots for different values of $\beta$ reveals that larger $\beta$ produces significantly narrower bound width. When $\beta=0.45$, point identification of the ATE$(x)$ is achieved for most of the $(\gamma,\rho)$ pairs because the condition in Proposition \ref{prop3_1} is satisfied.
\begin{sidewaysfigure}[p]
\centering
\caption{Manski and SV Bounds for ATE ($x=\mathbb{E}[X]$)}
\label{fig:3d}
\includegraphics[width=1.05\textwidth]{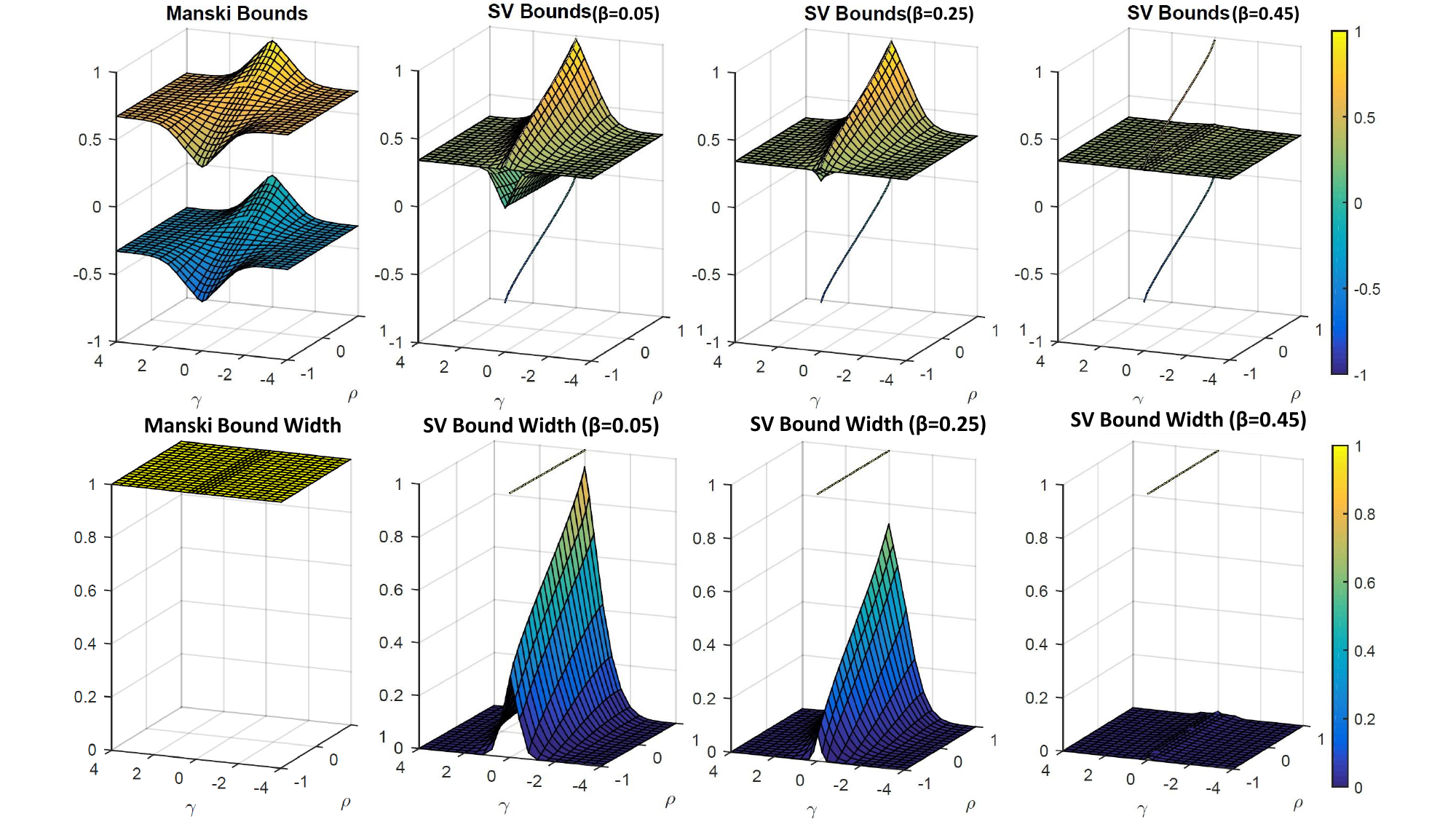}
\vspace{-0.1cm}
\footnotesize{Note: Three dimensional plots of the ATE bounds as function of $(\gamma,\rho)$. When $\gamma=0$, SV bounds reduce to Manski bounds with bound width one.}
\end{sidewaysfigure}

\subsection{Identification Gains Decomposition}
Figure \ref{fig:IG0.25} displays the decompositions of identification gains for $\gamma\in\{1,2\}$, $\rho\in\{-0.8,-0.5,0.5,0.8\}$ and $\beta\in\{0.05,0.25,0.45\}$. For the positive ATE$(x)$ case, the contribution of IV validity, $C_1(x)$, is determined by the Manski lower bound, and decreases as $\rho$ increases,\footnote{Conversely the numerical results not reported here show that when the ATE$(x)$ is negative $C_1(x)$ increases as $\rho$ increases.} while $C_1(x)$ is invariant to $\beta$.
By way of contrast, the component $C_2(x)$ also does not change by $\beta$, but it increases significantly as the magnitude of $\gamma$ increases. The component of identification gains due to the exogenous covariates, $C_3(x)$, also contributes significantly to the bounds. When $\beta$ is relatively large (e.g. $\beta=0.45$), point identification is virtually achieved.
\begin{figure}[p]
\begin{center}
\caption{Decomposition of Identification Gains ($x=\mathbb{E}[X]$)}\label{fig:IG0.25}
\vspace{-0.25cm}
\subcaption{$\beta=0.05$}
\vspace{-0.3cm}
\includegraphics[width=0.95\textwidth]{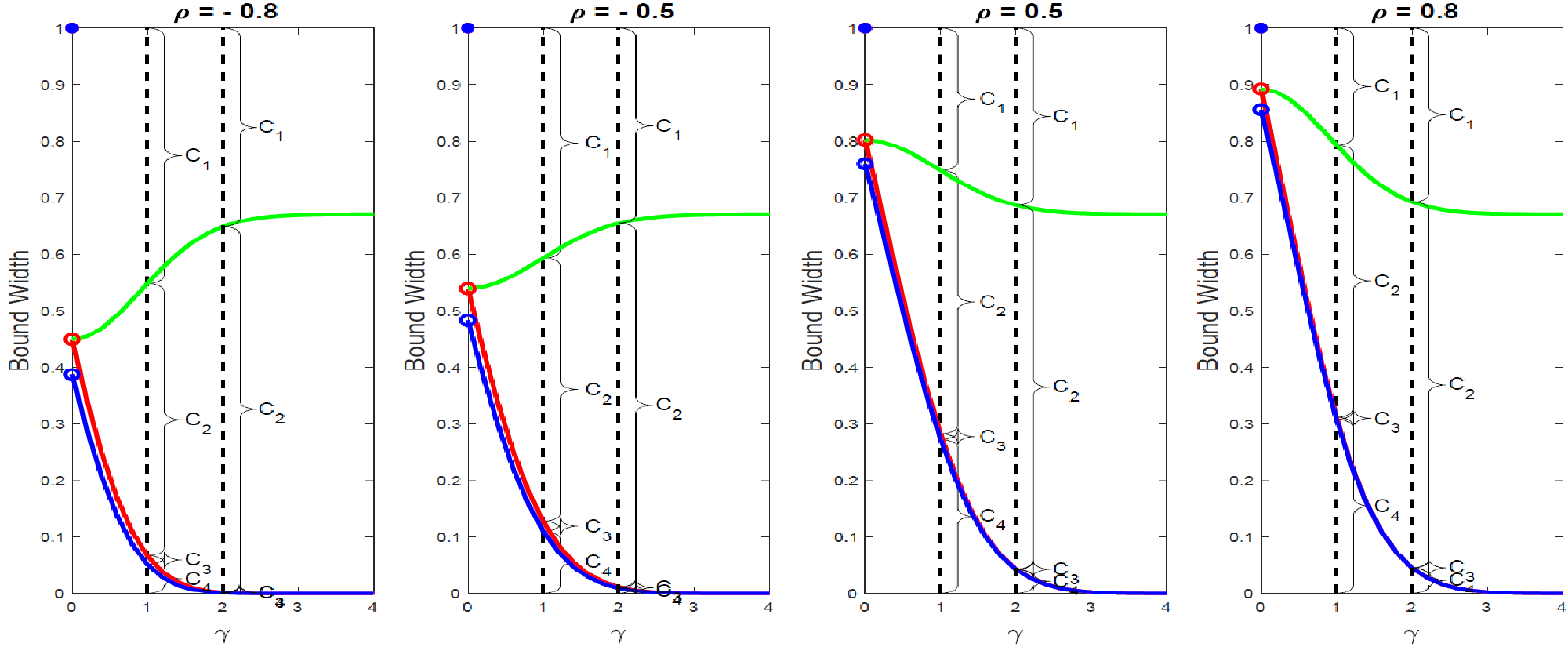}
\vspace{-0.3cm}
\subcaption{$\beta=0.25$}
\vspace{-0.3cm}
\includegraphics[width=0.95\textwidth]{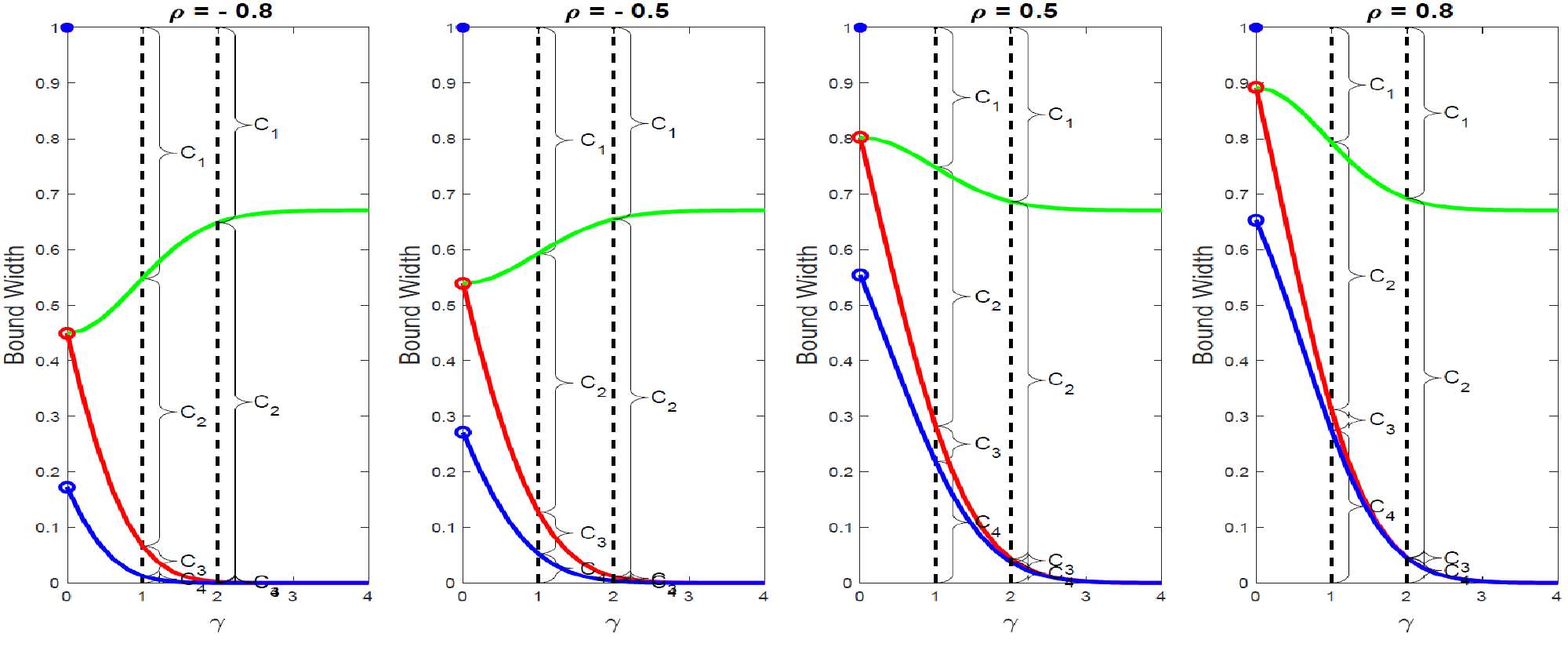}
\vspace{-0.3cm}
\subcaption{$\beta=0.45$}
\vspace{-0.3cm}
\includegraphics[width=0.95\textwidth]{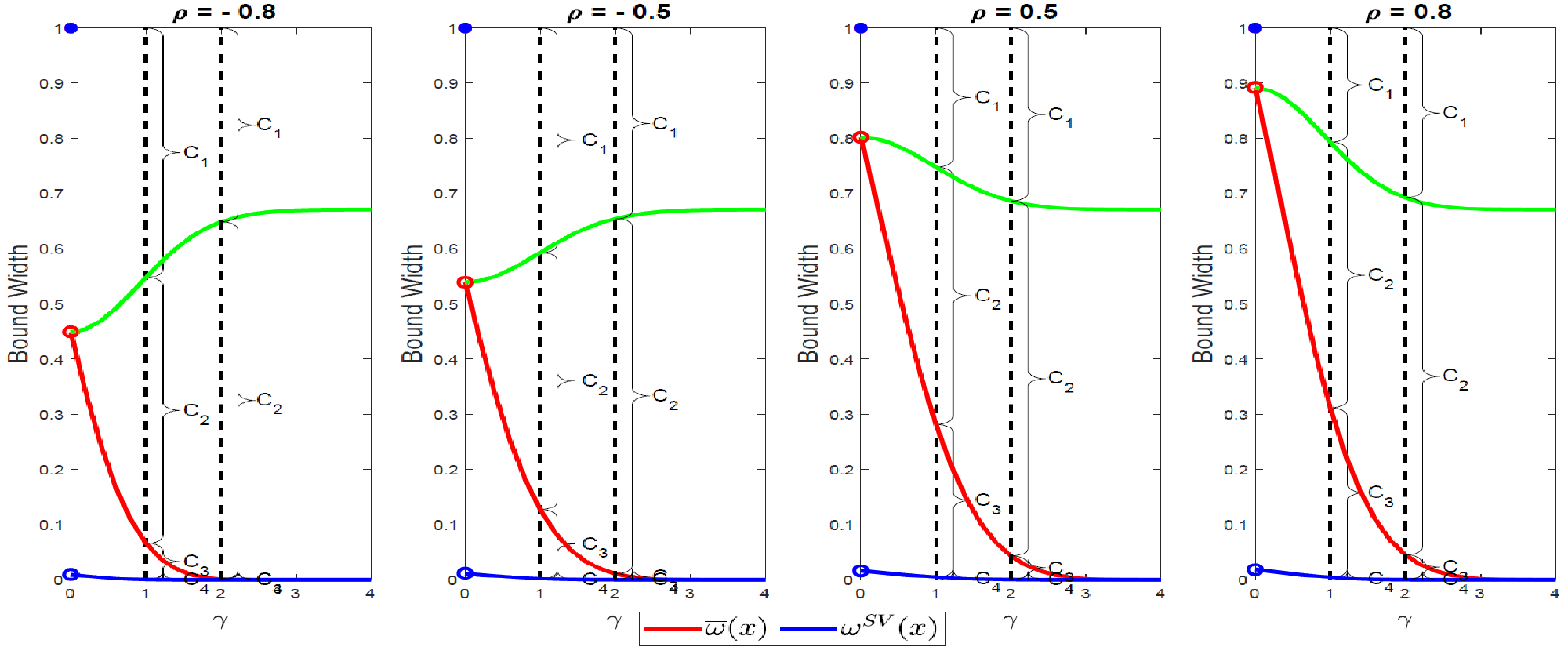}
\end{center}
\vspace{-0.3cm}
\footnotesize{Note: The green line depicts the amount of IV validity contribution $C_1(x)$. To aid legibility $C_1(x),\ldots,C_4(x)$ have been rendered as $C_1,\ldots,C_4$ in each of the subplots in this figure. x-axis displays the values of $\gamma$. For space limitation, we only represent the figure for nonnegative values of $\gamma$.}
\end{figure}

\subsection{IV Identification Power}
\begin{figure}[H]
\begin{center}
\caption{Instrument Identification Power ($x=\mathbb{E}[X]$)}
\vspace{-0.3cm}
\label{fig:R}
\includegraphics[width=0.5\textwidth]{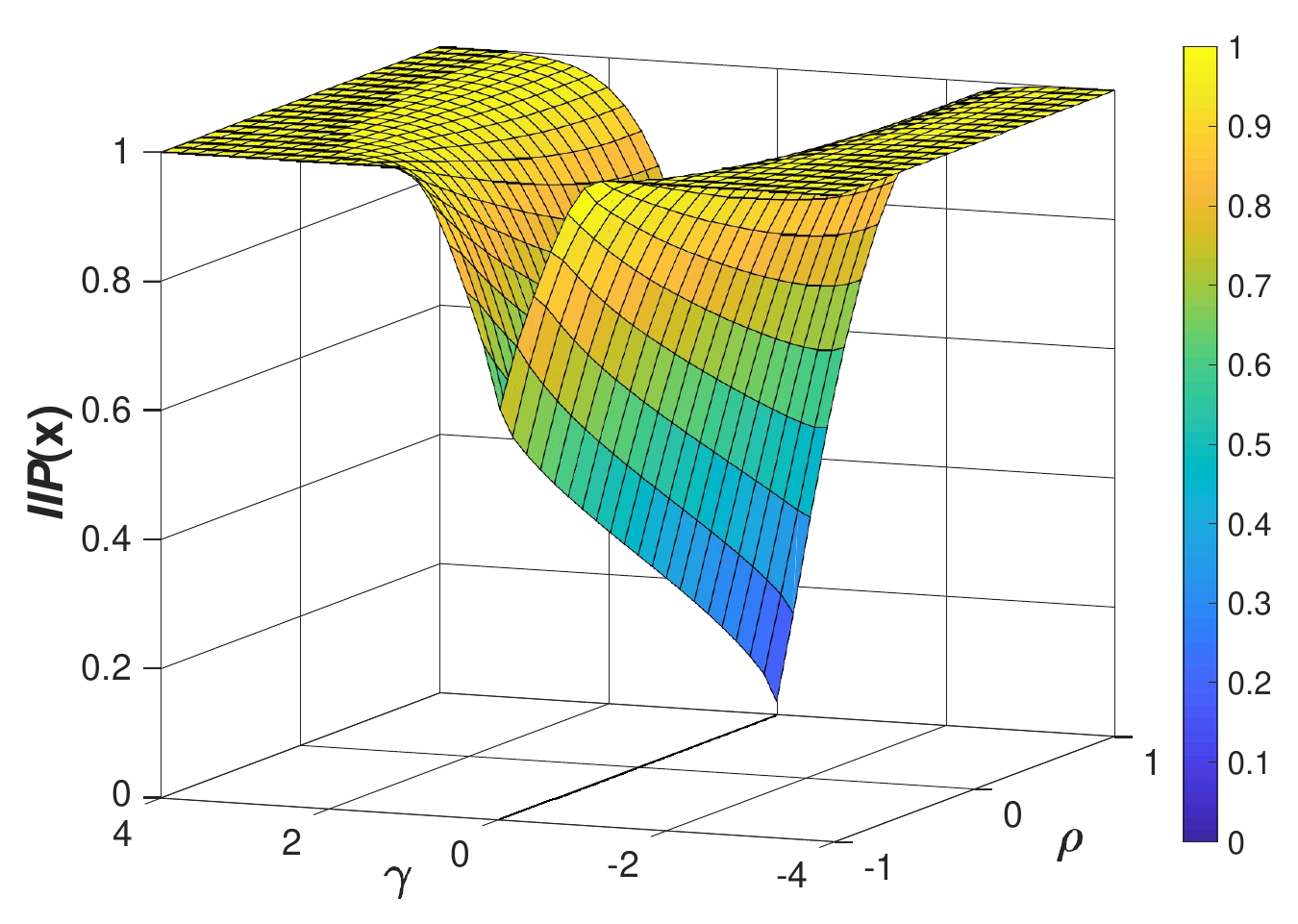}
\end{center}
\vspace{-0.1cm}
\footnotesize{Note: Three dimensional plot of $IIP(x)$ as function of $(\gamma,\rho)$. The value of $\beta$ does not affect the $IIP(x)$ in this case because $\pi=0$ and no matches of Pr$[D=1|x,z]=$Pr$[D=1|x',z']$ exist for $x=\mathbb{E}[X]$ and $z,z'\in\{-1,1\}$. When $\gamma=0$, the $IIP(x)=0$ because IV is irrelevant.}
\end{figure}

Figure \ref{fig:R} depicts the index $IIP(x)$ as a function of $(\gamma,\rho)$. The plot confirms that, firstly, $IIP(x)$ is bigger when IVs are stronger ($|\gamma|$ higher). In addition, for a given IV strength, higher $IIP(x)$ can be achieved if $\rho$ has an opposite sign from the ATE$(x)$ and is of high magnitude ($|\rho|$). If $\rho$ is of the same sign as the ATE$(x)$, then the lower the degree of endogeneity the better the identification power.\footnote{Due to space limitation, numerical results not displayed here show that under the DGP \eqref{model2}, our results regarding how the IV strength and its interaction with the direction and magnitude of endogeneity improve the benchmark Manski bounds also hold in other common ATE bounding analyses, including \citet{heckman2001instrumental} and \citet{chesher2010instrumental}.}

\section{Finite Sample Evaluation and Implications for Practice}\label{sectionER}
In this section, we first present finite samples results on the decomposition and $IIP(x)$. We also illustrate how $IIP(x)$ can be used to compare identification power of alternative IV sets, as there are empirical situations where selection among alternative IVs may be of interest. Second, we discuss practical implications of our results for applied researchers.
\subsection{Finite Sample Evaluation}
We first examine situations when irrelevant, incomplete, or mis-specified IVs are used. We design DGPs with the same set of IVs but with different directions and degrees of endogeneity and different covariant variations, so that the same IV strength can exert different identification powers. Consider i.i.d. samples drawn from DGP in \eqref{model2} with two IVs:
\begin{equation}\label{model3}
\begin{aligned}
Y&=1[\alpha D+\beta X+\varepsilon_1>0],\\
D&=1[\pi X+\gamma_1Z_1+\gamma_2Z_2+\varepsilon_2>0]
\end{aligned}
\end{equation}
where $Z_1$, $Z_2$ and $X$ are mutually independent, and also independent of $(\varepsilon_1,\varepsilon_2)$, $Z_1\sim Bernoulli(1/2)$ and $Z_2\in\{-3,-2,-1,0,1,2,3\}$ with probabilities $(0.1,0.1,0.2,0.2,0.2,0.1,0.1)$. Set $\alpha=1$, $\beta=1$, $\pi=-1$, $(\gamma_1,\gamma_2)=(0.5,0.2)$, and $(\varepsilon_1,\varepsilon_2)$ is jointly normal with mean zero, variance one and correlation $\rho\in\{0.5,0.8\}$. Consider two cases of covariate variability: $X\sim\mathbb{N}(0,1)$ and $X\sim Bernoulli(1/2)$, as in Table \ref{tab:MC_DGP_setting}. We focus on ATE$(x)$ with $x=0$.

\begin{table}[htb!]
\begin{center}
\caption{DGP Designs in Finite Sample Evaluation}\label{tab:MC_DGP_setting}
\setlength{\tabcolsep}{4pt}
\vspace{-0.25cm}
		\begin{tabular}{llcccccc}
			\hline
			\hline
&&&\multicolumn{2}{c}{DGP Designs}&&CPS extremes $(\underline{p}(x),\overline{p}(x))$&ATE$(x)$ ($x=0$)\cr
\hline
\multirow{2}{*}{case 1}&Table \ref{tab:QMLErho0.5case2}&&$X\sim\mathbb{N}(0,1)$&$\rho=0.5$&&\multirow{2}{*}{(0.274,0.864)}&\multirow{2}{*}{0.341}\cr
&Table \ref{tab:QMLErho0.8case2}&&$X\sim\mathbb{N}(0,1)$&$\rho=0.8$\cr
\cr
\multirow{2}{*}{case 2}&Table \ref{tab:QMLErho0.5casebinary}&&$X\sim Bernoulli(1/2)$&$\rho=0.5$&&\multirow{2}{*}{(0.274,0.864)}&\multirow{2}{*}{0.341}\cr
&Table \ref{tab:QMLErho0.8casebinary}&&$X\sim Bernoulli(1/2)$&$\rho=0.8$\cr
\hline
\hline
\end{tabular}
\end{center}
\vspace{-0.3cm}
\end{table}

We introduce two "pseudo" IVs: $\widetilde{Z}_2=1[Z_2>0]$, a misspecified binary IV that only partially reflects $Z_2$, and an irrelevant IV $Z_3\in\{0,1\}$ such that $\text{Pr}[Z_3=1]=2/3$, and $Z_3\perp(\varepsilon_1,\varepsilon_2,Z_1,Z_2,X)$. To evaluate the finite sample performance of $IIP(x)$ and selection of IVs, consider five alternative sets of IVs: (1) only $Z_1$; (2) only $Z_2$; (3) $Z_1$ and misspecified $\widetilde{Z}_2$; (4) $Z_1$ and $Z_2$; and (5) $(Z_1,Z_2)$ and irrelevant $Z_3$.

\begin{table}[htb!]
\begin{center}
\caption{Population CPS Range and $IIP(x)$ ($x=0$, cases 1 and 2)}\label{tab:range}
\vspace{-0.25cm}
		\begin{tabular}{cllccc}
			\hline
			\hline
&IVs&CPS definition&CPS extremes $(\underline{p}(x),\overline{p}(x))$&$IIP(x)$ & $IIP(x)$ \cr
&&&&($\rho=0.5$)&($\rho=0.8$)\cr
\hline
(1)&only $Z_1$&Pr$[D=1|x,Z_1]$&$(0.500,0.682)$&0.305&0.232\cr
(2)&only $Z_2$&Pr$[D=1|x,Z_2]$&$(0.367,0.795)$&0.493&0.443\cr
(3)&$Z_1,\widetilde{Z}_2$&Pr$[D=1|x,Z_1,\widetilde{Z}_2]$&$(0.410,0.799)$&0.456&0.403\cr
(4)&$Z_1,Z_2$&Pr$[D=1|x,Z_1,Z_2]$&$(0.274,0.864)$&0.625&0.594\cr
(5)&$Z_1,Z_2,Z_3$&Pr$[D=1|x,Z_1,Z_2,Z_3]$&$(0.274,0.864)$&0.625&0.594\cr
\hline
\hline
\end{tabular}
\end{center}
\vspace{-0.3cm}\footnotesize
Note: The population CPS and $IIP(x)$ are the same for case 1 and case 2 with difference covariate variability.
\end{table}

Table \ref{tab:range} presents the CPS extremes and $IIP(x)$ of the five IV sets under the true DGP, which are the same for case 1 and 2 because the covariate variability does not impact them. We can see that the $(\underline{p}(x),\overline{p}(x))$ is the widest in (4) when both $Z_1$ and $Z_2$ are used. Adding an irrelevant $Z_3$ does not change the $(\underline{p}(x),\overline{p}(x))$, so theoretically (5) has the same IV strength as (4). The $(\underline{p}(x),\overline{p}(x))$ shrinks when only one valid IV is used as in (1) or (2). As expected, when a valid IV is incorrectly specified as a proxy dummy $\widetilde{Z}_2$ in (3), the $(\underline{p}(x),\overline{p}(x))$ is narrower than that of the best set in (4), but wider than that in (1) with $Z_1$ alone. Interestingly, comparing IV set (3) with (2), set (2) with only one valid IV actually results in wider $(\underline{p}(x),\overline{p}(x))$ than that for the two IVs in set (3) with $Z_2$ misspecified.

Whilst the CPS range indicates the IV strength, it is the $IIP(x)$ that captures the identification power of each IV set, measuring the reduction of SV bound width relative to the benchmark Manski bound width due to the contribution of IVs. As seen from the two $IIP(x)$ columns in Table \ref{tab:range}, the same IV strength can achieve bigger identification gains for $\rho=0.5$ than that with $\rho=0.8$. This is consistent with the results in Section \ref{sectionNA}: as $\rho$ and ATE$(x)$ are both positive in this case, the lower absolute value of $\rho$, the higher the $IIP(x)$ is. For example for IV set (4), the Manski bound width can be reduced by $59.4\%$ by the two IVs when $\rho=0.8$, and it increases to $62.5\%$ if $\rho=0.5$.
For given $\rho$, the equally most powerful IV sets are (4) and (5), and the least powerful set is (1).

We next present the finite sample estimation of the Manski and SV bounds, and conduct the decomposition analysis. Sample size is $n=500,5000,10000$ and replicate $M=1000$ times. Tables \ref{tab:QMLE_normal} to \ref{tab:QMLE_binary} present the sample average (over $M$ replications) of the estimated bounds, $C_{1}(x)$ to $C_{4}(x)$ and $IIP(x)$ of the five IV sets at $x=0$. We use the ``half-median-unbiased estimator'' (HMUE) of the intersecting bounds proposed by \citet*{chernozhukov2013intersection} (hereafter CLR) to estimate the benchmark Manski bounds and the SV bounds. In particular, we employ maximum likelihood estimation (MLE) to estimate the bounding functions and to select the critical values for bias correction according to the simulation-based methodology of CLR.\footnote{We report the HMUE of the Manski bounds, for comparison purpose. Other estimation methods for Manski bounds are also available, see e.g. \citet{imbens2004confidence}. See Appendix \ref{app_CLR} for more details about the CLR method.}

Tables \ref{tab:QMLE_normal} and \ref{tab:QMLE_binary} present results under two different covariate distributions. The first row in each table lists the ATE bounds and decomposition components under the true DGP. We can see that in \textbf{case 1} (Table \ref{tab:QMLE_normal}), where $X$ possesses sufficient variation, the true SV bounds point identify the ATE$(x)$ for both $\rho=0.5$ and $\rho=0.8$. In \textbf{case 2} (Table \ref{tab:QMLE_binary}), the true SV bounds fail to point identify the ATE$(x)$ due to the limited variation in $X$.

Next, we focus on the left part of each table, which displays the HMUEs of the ATE bounds, and the Hausdorff distance between the true and estimated bounds, at $x=0$.\footnote{Simulation results at different values $x$ display similar patterns to those at $x=0$, therefore are not reported due to the space limitation. 
Hausdorff distance has been employed to study convergence properties when a set is the parameter of interest, see e.g. \citet{manski2002inference}.} For all four tables, we can see that the estimated Manski bounds are the same across all five IV sets, always include zero, and have a width a little over one. The estimated SV bounds identify the sign of ATE$(x)$ for all five IV sets. Moreover, the IV sets with greater identification power lead to narrower estimated SV bounds and also improve the estimation accuracy in most of the scenarios. More precisely, the Hausdorff distance of the estimated SV bounds to the true bounds decreases as the IV identification power increases.

Moving to the right part of each of table, first, we note that for each given IV set, all the estimated $C_1(x)$ to $C_4(x)$ and $IIP(x)$ converges to their true values as sample size $n$ increases, 
indicating that the estimated identification gain is more accurate for larger sample size.\footnote{\label{footnote}Because $C_1(x)$ to $C_4(x)$ are functions of $L^M(x)$, $U^M(x)$, $\overline{\omega}(x)$ and $\omega^{SV}(x)$, the estimates of $C_1(x)$ to $C_4(x)$ are computed using the HMUE of the bounds or their widths. We compute $\overline{\omega}(x)$ as the width of the estimated bounds (by HMUE of CLR) $[\underline{L}^{SV}(x),\overline{U}^{SV}(x)]$ in \eqref{svw_p} if ATE$(x)>0$ is identified, or \eqref{svw_n} if otherwise.} We also note that the estimated $C_1(x)$ is the same for different IV sets. This result is intuitive because the identification gains brought by the IV validity should not vary with the IV strength. Comparison of Tables \ref{tab:QMLErho0.5case2} and \ref{tab:QMLErho0.8case2} or Tables \ref{tab:QMLErho0.5casebinary} and \ref{tab:QMLErho0.8casebinary} reveals that the impacts of endogeneity degree on IV identification power can be captured by the estimated $IIP(x)$. Importantly, the true ranking of $IIP(x)$ as in Table \ref{tab:range} can be correctly revealed by finite sample estimates of $IIP(x)$.

\FloatBarrier 
\begin{table}[p!]
\setlength{\tabcolsep}{3.5pt}
\begin{center}
\caption{\textbf{Case 1}. True and Estimated Bounds, and Decomposition of Identification Gains}\label{tab:QMLE_normal} \subcaption{$\rho=0.5,~X\sim\mathbb{N}(0,1),~x=0$}\label{tab:QMLErho0.5case2}
\begin{small}
		\begin{tabular}{clccccccccccc}
			\hline
			\hline
&&\multicolumn{4}{c}{Bounds}&&\multicolumn{5}{c}{Decomposition}\cr
\cline{3-6}\cline{8-12}
\multirow{2}{*}{}&&\multicolumn{2}{c}{Manski}&\multicolumn{2}{c}{SV}&&\cr
&&\footnotesize$[L^M(x),U^M(x)]$&$d_{H}(x)$&\footnotesize$[L^{SV}(x),U^{SV}(x)]$&$d_{H}(x)$&&$C_1(x)$&$C_2(x)$&$C_3(x)$&$C_4(x)$&IIP(x)\cr
\hline
\small True DGP&~~~~$Z_1,Z_2$&$[-0.179,0.821]$&&$[0.341,0.341]$&&&0.179& 0.446& 0.375 & 0.000 & 0.625 \\
\hline
\multirow{5}{*}{\small$n=500$}
&\small (1) only $Z_1$&\multirow{5}{*}{[-0.246,0.899]}&\multirow{5}{*}{0.092}&[0.117,0.775]&0.434&&\multirow{5}{*}{0.246} & 0.186 & 0.056 & 0.658 & 0.432\\
&\small (2) only $Z_2$&&&[0.246,0.562]& 0.227&&& 0.342 & 0.241 & 0.316 & 0.587\\
&\small (3) $Z_1,\widetilde{Z}_2$&&&[0.193,0.759]& 0.418&&& 0.218 & 0.116 & 0.565 & 0.464\\
&\small (4) $Z_1,Z_2$&&&[0.290,0.455]& 0.121&&& 0.436 & 0.298 & 0.165 & 0.682\\
&\small (5) $Z_1,Z_2,Z_3$&&&[0.300,0.451]& 0.116&&& 0.424 & 0.324 & 0.151 & 0.670\\
\hline
\multirow{5}{*}{\small$n=5000$}
&\small (1) only $Z_1$&\multirow{5}{*}{[-0.202,0.846]}&\multirow{5}{*}{0.030}&[0.121,0.768]& 0.427&&\multirow{5}{*}{0.202}& 0.145 & 0.053 & 0.648 & 0.347\\
&\small (2) only $Z_2$&&&[0.266,0.372]& 0.078& && 0.334 & 0.406 & 0.106 & 0.536\\
&\small (3) $Z_1,\widetilde{Z}_2$&&&[0.221,0.757]& 0.416&&& 0.194 & 0.116 & 0.536 & 0.395\\
&\small (4) $Z_1,Z_2$&&&[0.312,0.377]& 0.043&&& 0.446 & 0.335 & 0.066& 0.648\\
&\small (5) $Z_1,Z_2,Z_3$&&&[0.316,0.373]& 0.038&&&0.442 &0.347 & 0.057& 0.644\\
\hline
\multirow{5}{*}{\small$n=10000$}
&\small (1) only $Z_1$&\multirow{5}{*}{[-0.198,0.838]}&\multirow{5}{*}{0.022}&[0.123,0.768]& 0.427&&\multirow{5}{*}{0.198} & 0.139 & 0.054 & 0.645 & 0.337\\
&\small (2) only $Z_2$&&&[0.263,0.363]& 0.080 &&& 0.331 & 0.407 & 0.101 & 0.528\\
&\small (3) $Z_1,\widetilde{Z}_2$&&&[0.225,0.756]& 0.414&&& 0.189 & 0.118 & 0.531& 0.387\\
&\small (4) $Z_1,Z_2$&&&[0.317,0.365]& 0.031&&& 0.444 & 0.346 & 0.048 & 0.642 \\
&\small (5) $Z_1,Z_2,Z_3$&&&[0.320,0.362]& 0.027&&&0.443&0.353&0.042 &0.641 \\
\hline
\hline
\end{tabular}
\end{small}
%
\vspace{0.5cm}
\subcaption{$\rho=0.8,~X\sim\mathbb{N}(0,1),~x=0$}\label{tab:QMLErho0.8case2}
\begin{small}
		\begin{tabular}{clccccccccccc}
			\hline
			\hline
&&\multicolumn{4}{c}{Bounds}&&\multicolumn{5}{c}{Decomposition}\cr
\cline{3-6}\cline{8-12}
\multirow{2}{*}{}&&\multicolumn{2}{c}{Manski}&\multicolumn{2}{c}{SV}&&\cr
&&\footnotesize$[L^M(x),U^M(x)]$&$d_{H}(x)$&\footnotesize$[L^{SV}(x),U^{SV}(x)]$&$d_{H}(x)$&&$C_1(x)$&$C_2(x)$&$C_3(x)$&$C_4(x)$&IIP(x)\cr
\hline
\small True DGP&~~~~$Z_1,Z_2$&$[-0.096,0.904]$&&$[0.341,0.341]$&&&0.096& 0.498 & 0.406 & 0.000 & 0.594\\
\hline
\multirow{5}{*}{\small$n=500$}
&\small (1) only $Z_1$&\multirow{5}{*}{[-0.157,0.996]}&\multirow{5}{*}{0.098}&[0.124,0.873]&0.532&&\multirow{5}{*}{0.157}&0.205&0.041&0.750& 0.362\\
&\small (2) only $Z_2$&&&[0.233,0.559]& 0.229&&&0.382&0.288&0.326&0.539\\
&\small (3) $Z_1,\widetilde{Z}_2$&&&[0.191,0.848]& 0.507&&& 0.246&0.093&0.657&0.403\\
&\small (4) $Z_1,Z_2$&&&[0.291,0.437]& 0.107&&&0.495&0.355&0.146&0.652\\
&\small (5) $Z_1,Z_2,Z_3$&&&[0.298,0.431]&0.100&&&0.482&0.382& 0.133&0.639\\
\hline
\multirow{5}{*}{\small$n=5000$}
&\small (1) only $Z_1$&\multirow{5}{*}{[-0.121,0.924]}&\multirow{5}{*}{0.028}&[0.128,0.860]& 0.519&&\multirow{5}{*}{0.121}&0.149&0.042&0.732&0.271\\
&\small (2) only $Z_2$&&&[0.254,0.357]& 0.088& &&0.346&0.475&0.103&0.467\\
&\small (3) $Z_1,\widetilde{Z}_2$&&&[0.208,0.853]& 0.512&&&0.210&0.068&0.645&0.332 \\
&\small (4) $Z_1,Z_2$&&&[0.312,0.378]& 0.043&&&0.489&0.369&0.066&0.610\\
&\small (5) $Z_1,Z_2,Z_3$&&&[0.315,0.373]& 0.038&&&0.486&0.380&0.058&0.607\\
\hline
\multirow{5}{*}{\small$n=10000$}
&\small (1) only $Z_1$&\multirow{5}{*}{[-0.117,0.918]}&\multirow{5}{*}{0.022} &[0.129,0.860]& 0.519&&\multirow{5}{*}{0.117}&0.146 & 0.042 & 0.731 & 0.263\\
&\small (2) only $Z_2$&&&[0.258,0.357]& 0.083 &&& 0.346 & 0.473 & 0.099 & 0.463\\
&\small (3) $Z_1,\widetilde{Z}_2$&&&[0.212,0.851]& 0.510&&& 0.209& 0.071 & 0.639 & 0.326\\
&\small (4) $Z_1,Z_2$&&&[0.316,0.369]& 0.034&&& 0.491 & 0.374 & 0.053 & 0.607\\
&\small (5) $Z_1,Z_2,Z_3$&&&[0.319,0.365]& 0.030&&& 0.491 & 0.381 & 0.046 & 0.607\\
\hline
\hline
\end{tabular}
\end{small}
\end{center}
{\footnotesize{Note: The estimated bounds, the Hausdorff distance $d_H(x)$ and the decompositions are the averages over 1000 replications.}}
\end{table}

\begin{table}[p]
\setlength{\tabcolsep}{3.5pt}
\begin{center}
\caption{\textbf{Case 2}. True and Estimated Bounds, and Decomposition of Identification Gains}\label{tab:QMLE_binary}
\subcaption{$\rho=0.5,~X\sim Bernoulli(1/2),~x=0$}\label{tab:QMLErho0.5casebinary}
\begin{small}
		\begin{tabular}{clccccccccccc}
			\hline
			\hline
&&\multicolumn{4}{c}{Bounds}&&\multicolumn{5}{c}{Decomposition}\cr
\cline{3-6}\cline{8-12}
\multirow{2}{*}{}&&\multicolumn{2}{c}{Manski}&\multicolumn{2}{c}{SV}&&\cr
&&\footnotesize$[L^M(x),U^M(x)]$&$d_{H}(x)$&\footnotesize$[L^{SV}(x),U^{SV}(x)]$&$d_{H}(x)$&&$C_1(x)$&$C_2(x)$&$C_3(x)$&$C_4(x)$&IIP(x)\cr
\hline
\small True DGP&~~~~$Z_1,Z_2$&$[-0.179,0.821]$&&$[0.283,0.547]$&&&0.179& 0.446 & 0.111 & 0.264& 0.625\\
\hline
\multirow{5}{*}{\small$n=500$}
&\small (1) only $Z_1$&\multirow{5}{*}{[-0.263,0.904]}&\multirow{5}{*}{0.102}&[0.060,0.776]&0.237&&\multirow{5}{*}{0.263} &0.185&0.002&0.716&0.448\\
&\small (2) only $Z_2$&&&[0.110,0.669]& 0.179&&&0.359& 0.009& 0.559& 0.621\\
&\small (3) $Z_1,\widetilde{Z}_2$&&&[0.098,0.769]& 0.224&&& 0.237& 0.000&0.671&0.499\\
&\small (4) $Z_1,Z_2$&&&[0.166,0.647]& 0.131&&&0.439& 0.003& 0.481& 0.701\\
&\small (5) $Z_1,Z_2,Z_3$&&&[0.160,0.656]& 0.140&&&0.433& 0.001& 0.496&0.695\\
\hline
\multirow{5}{*}{\small$n=5000$}
&\small (1) only $Z_1$&\multirow{5}{*}{[-0.206,0.849]}&\multirow{5}{*}{0.034}&[0.068,0.769]& 0.223&&\multirow{5}{*}{0.206}&0.148&0.000&0.701&0.354\\
&\small (2) only $Z_2$&&&[0.135,0.640]& 0.148& &&0.337& 0.010& 0.506& 0.543\\
&\small (3) $Z_1,\widetilde{Z}_2$&&&[0.115,0.754]&0.207&&&0.211& 0.000&0.639& 0.417\\
&\small (4) $Z_1,Z_2$&&&[0.210,0.619]& 0.079&&&0.446&0.002& 0.409& 0.653\\
&\small (5) $Z_1,Z_2,Z_3$&&&[0.208,0.620]& 0.081&&&0.444& 0.002& 0.412&0.650\\
\hline
\multirow{5}{*}{\small$n=10000$}
&\small (1) only $Z_1$&\multirow{5}{*}{[-0.198,0.841]}&\multirow{5}{*}{0.024} &[0.069,0.768]& 0.221&&\multirow{5}{*}{0.198}&0.141&0.001&0.699&0.339\\
&\small (2) only $Z_2$&&&[0.138,0.640]& 0.145 &&&0.333&0.007& 0.502& 0.531\\
&\small (3) $Z_1,\widetilde{Z}_2$&&&[0.118,0.751]& 0.204&&& 0.207& 0.000&0.633&0.406\\
&\small (4) $Z_1,Z_2$&&&[0.216,0.612]& 0.070&&&0.447& 0.000& 0.396& 0.645\\
&\small (5) $Z_1,Z_2,Z_3$&&&[0.217,0.613]& 0.071&&&0.447& 0.002& 0.396&0.645\\
\hline
\hline
\end{tabular}
\end{small}
%
\vspace{0.5cm}
\subcaption{$\rho=0.8,~X\sim Bernoulli(1/2),~x=0$}\label{tab:QMLErho0.8casebinary}
\begin{small}
		\begin{tabular}{clccccccccccc}
			\hline
			\hline
&&\multicolumn{4}{c}{Bounds}&&\multicolumn{5}{c}{Decomposition}\cr
\cline{3-6}\cline{8-12}
\multirow{2}{*}{}&&\multicolumn{2}{c}{Manski}&\multicolumn{2}{c}{SV}&&\cr
&&\footnotesize$[L^M(x),U^M(x)]$&$d_{H}(x)$&\footnotesize$[L^{SV}(x),U^{SV}(x)]$&$d_{H}(x)$&&$C_1(x)$&$C_2(x)$&$C_3(x)$&$C_4(x)$&IIP(x)\cr
\hline
\small True DGP&~~~~$Z_1,Z_2$&$[-0.096,0.904]$&&$[0.319,0.593]$&&&0.096& 0.498 & 0.132 & 0.274 & 0.594\\
\hline
\multirow{5}{*}{\small$n=500$}
&\small (1) only $Z_1$&\multirow{5}{*}{[-0.165,0.972]}&\multirow{5}{*}{0.084}&[0.077,0.868]&0.276&&\multirow{5}{*}{0.165}& 0.183&0.000&0.790&0.348\\
&\small (2) only $Z_2$&&&[0.114,0.751]& 0.212&&&0.330&0.018& 0.637& 0.495\\
&\small (3) $Z_1,\widetilde{Z}_2$&&&[0.133,0.863]& 0.270&&& 0.243& 0.000&0.730&0.408\\
&\small (4) $Z_1,Z_2$&&&[0.209,0.732]& 0.154&&&0.458& 0.008& 0.523& 0.623\\
&\small (5) $Z_1,Z_2,Z_3$&&&[0.200,0.738]& 0.164&&&0.441& 0.009& 0.538&0.606\\
\hline
\multirow{5}{*}{\small$n=5000$}
&\small (1) only $Z_1$&\multirow{5}{*}{[-0.117,0.925]}&\multirow{5}{*}{0.026}&[0.086,0.861]& 0.268&&\multirow{5}{*}{0.117}&0.149&0.001&0.776&0.266\\
&\small (2) only $Z_2$&&&[0.144,0.720]&0.175&&&0.340& 0.010& 0.576& 0.457\\
&\small (3) $Z_1,\widetilde{Z}_2$&&&[0.154,0.848]& 0.256&&& 0.232& 0.000&0.694&0.349\\
&\small (4) $Z_1,Z_2$&&&[0.255,0.694]& 0.102&&&0.486& 0.006& 0.439& 0.603\\
&\small (5) $Z_1,Z_2,Z_3$&&&[0.255,0.696]& 0.105&&&0.483& 0.013& 0.440& 0.600\\
\hline
\multirow{5}{*}{\small$n=10000$}
&\small (1) only $Z_1$&\multirow{5}{*}{[-0.111,0.919]}&\multirow{5}{*}{0.019} &[0.087,0.860]& 0.267&&\multirow{5}{*}{0.111}&0.146&0.000&0.773&0.257\\
&\small (2) only $Z_2$&&&[0.148,0.713]& 0.171 &&&0.338& 0.017& 0.565& 0.450\\
&\small (3) $Z_1,\widetilde{Z}_2$&&&[0.158,0.846]& 0.253&&& 0.230& 0.000&0.688&0.342\\
&\small (4) $Z_1,Z_2$&&&[0.263,0.693]& 0.100&&&0.491&0.001& 0.430& 0.603\\
&\small (5) $Z_1,Z_2,Z_3$&&&[0.263,0.692]& 0.100&&&0.489& 0.003& 0.429& 0.601\\
\hline
\hline
\end{tabular}
\end{small}
\end{center}
{\footnotesize{Note: The estimated bounds, the Hausdorff distance $d_H(x)$ and the decompositions are the averages over 1000 replications.}}
\end{table}

It is interesting to analyze the effect of adding an additional but completely irrelevant IV on the finite sample performance of ATE partial identification, by comparing the results obtained using IV sets (4) and (5). Adding $Z_3$ to $(Z_1,Z_2)$ actually produces a small \textit{decrease} of the estimated $IIP(x)$, on average, for almost all different DGP designs considered in this section. 
The Cramer-Von Mises test and the Kolmogorov–Smirnov test confirm that the average values of the estimates of $IIP(x)$ under scenario (4) are significantly different from those obtained under scenario (5), when sample size is $N=500$ and $N=5000$ for both endogeneity degrees and for both case 1 and case 2. While when sample size is sufficiently large $N=10000$, the estimates of $IIP(x)$ under scenario (4) and (5) are no longer significantly different, except for case 2 with $\rho=0.8$. Particularly, from Table \ref{tab:QMLE_binary} we can see that when $X$ is a binary variable, on average, the estimated SV bounds using $(Z_1,Z_2)$ are narrower than those estimated by the IV set including the irrelevant IV, especially for small sample size. Analyzing the results across the replications in case 2, we find that about 78\% (for both endogeneity degrees) of the replications give narrower estimated SV bounds with IV set $(Z_1,Z_2)$ than those with $(Z_1,Z_2,Z_3)$, for sample size $N=500$; and this rate becomes to 53\% ($\rho=0.5$) and 64\% ($\rho=0.8$) for sufficiently large sample size $N=10000$. This suggests that in finite sample, the loss of information (efficiency) that arises from using irrelevant IV could lead to wider ATE bounds, especially when the covariate possesses limited variation.

On the other hand, the IV irrelevancy cannot always be detected by simply comparing the estimated SV bound width under different IV sets. That is, adding an irrelevant IV in (5) could further shrink the SV bound width when the covariate $X$ is continuous, although the improvement happens at the third decimal and the degree of the improvement decreases as sample size increases.
\footnote{In case 1, the shrinkage of the estimated SV bounds using the irrelevant $Z_3$ is due to the finite sample estimation error. In particular, because the estimates of the coefficient of the irrelevant $Z_3$ will be nonzero with probability one, it results in more matched pairs of $(x,z)$ and $(x',z')$ such that $|$Pr$[D=1|x,z]-$Pr$[D=1|x',z']|<c$ (see Appendix \ref{app_CLR}) especially when covariate is continuous. For case 1 in Table \ref{tab:QMLE_normal}, we find that when sample size is $N=500$, (i) there are 22\% ($\rho=0.5$) and 17\% ($\rho=0.8$) of the 1000 replications where at least one (either lower or upper) estimated SV bound using $(Z_1,Z_2)$ is closer to its true value, compared to that obtained by using the irrelevant IV; and (ii) 12\% of the replications yield wider estimated SV bounds when using the irrelevant IV, for both endogeneity degrees.}

\subsection{Implications for Practice}
Our analysis has several useful implications for policy design, choice of IVs for cost-effective identification targets, and assessing IV identification power in ATE bounding analysis.

First, as greater span of IV-driven propensity scores tightens the ATE bounds, increasing such span by creating IVs with stronger strength can be a consideration in designing policy instruments. For example, random allocation of treatment incentives or vouchers is often used as IV in studies with imperfect compliance. If there is a choice between two alternative IVs for encouraging treatment uptake with similar costs, the more segmented instrument design with both no incentives and high incentives to even a small proportion of participants, which leads to wider $(\underline{p}(x),\overline{p}(x))$, is preferable, compared to the less segmented ones with narrower span of the CPS. In addition, if an IV is continuous, it may not be a good practice to discretize it for simplicity purposes.

Second, for a specific analysis, sometimes the lower ATE bound being above a certain threshold, for example, zero, is more important for making the policy decision, whilst at other times, narrower ATE bounds may be more desirable. A priori knowledge, from economic theory or pilot experiments, on whether the direction of endogeneity is of the same sign as the ATE can be informative on the choice of identification target. When the ATE and the direction of endogeneity have opposite signs, the relatively narrow ATE bounds can be the target, because it is not too demanding of the IV strength and is sufficient to make policy recommendations. In contrast, if they appear to have the same sign, it is possible that increasing IV strength may only have limited scope for reducing ATE bound width. In this case, the researcher may instead want to focus on if the ATE lower bound exceeds a meaningful threshold. The attempt to create or collect stronger IVs in future studies may not be worthwhile, because stronger IVs would not help too much in improving the bounds.

In addition, for a given treatment variable, when there are multiple outcomes of interest, the preferable instrument design and identification target can be different, because the sign of the ATE and the direction of the treatment endogeneity may change with the choice of outcome.

Lastly, our simulation results show that in finite sample, using irrelevant IVs can result in a loss of efficiency and wider ATE bounds, especially when the covariate has limited variation. Our proposed $IIP$ index could be useful as a measure for detecting irrelevancy. Moreover, we reinforce \textit{a-fortiori} the warning that simply adding extra IVs without assessing their relevance is unlikely to be a good practical strategy.

\section{Empirical Application: Women LFP and Childbearing}\label{emp}
In this section, we apply our novel decomposition and IV evaluation method to study the effects of childbearing on women's labor supply. The dataset analyzed here is from the 1980 Census Public Use Micro Samples (PUMS), available at \citet{DVN/4W9GW2_2009}. We follow the data construction in \citet{angrist1998children}, where the sample consists of married women aged 21-35 with two or more children. The dateset contains 254,652 observations.
The outcome is an indicator of being paid for work in the year prior to the census. The treatment is an indicator of having more than two children.

Following \citet[][Table 11]{angrist1998children}, we use as continuous regressors woman's age, woman's age at first birth, and ages of the first two children (quarters), and binary regressors for first child being a boy, second child being a boy, black, hispanic, and other race, as well as the interactions of the above mentioned continuous and indicator variables. For computational simplicity, we reduce dimension of covariates by utilizing the estimated propensity score conditional on $X$, $X_P:=\widehat{Pr}[D=1|X]$ as a covariate, which is estimated via a probit model and $X$ includes all of the regressors mentioned above. Three sets of IVs are considered: (1) the binary indicator that the first two children are the same sex (``\emph{Samesex}''), (2) the binary indicator that the second birth is a twin (``\emph{Twins}''), and (3) both indicators (``\emph{Both}=\{\emph{Samesex},\emph{Twins}\}''). To provide a comparison of SV bounds with other commonly used ATE bounding analyses, we also compute the ATE bounds in \citet{heckman2001instrumental} (HV bounds) and \citet{chesher2010instrumental} (Chesher bounds). To be consistent with our previous numerical analyses in Section \ref{sectionER}, we use the method of CLR to compute all the four bounds of interest.

\begin{table}[ht!]
\begin{center}
\caption{Average of the Estimated Bounds}\label{tab:bounds}
\vspace{-0.35cm}
\subcaption{IV: Samesex}
\centering
		\begin{tabular}{lccccc}
			\hline
			\hline
$$&&Manski&HV&Chesher&SV\cr
\hline
\small HMUE &&[-0.560,0.439]&[-0.537,0.402]&[-0.537,-0.011] $\cup$ [0.011,0.402]&[-0.537,-0.031]\\
\small $90\%$ CI&&[-0.566,0.445]&[-0.546,0.410]&[-0.546,-0.005] $\cup$ [0.005,0.410]&[-0.546,-0.024]\\
\small $95\%$ CI&&[-0.567,0.446]&[-0.547,0.412]&[-0.547,-0.003] $\cup$ [0.003,0.412]&[-0.548,-0.023]\\
\small $99\%$ CI&&[-0.570,0.449]&[-0.551,0.416]&[-0.551,-0.001] $\cup$ [0.001,0.416]&[-0.551,-0.022]\\
			\hline
			\hline
		\end{tabular}
\subcaption{IV: Twins}
\centering
\begin{tabular}{lccccc}
			\hline
			\hline
&&Manski&HV&Chesher&SV\cr
\hline
\small HMUE &&[-0.560,0.439]&[-0.305,0.118]&[-0.305,-0.057]&[-0.182,-0.095]\\
\small $90\%$ CI&&[-0.566,0.445]&[-0.348,0.161]&[-0.348,-0.016]&[-0.263,-0.032]\\
\small $95\%$ CI&&[-0.567,0.446]&[-0.356,0.170]&[-0.356,-0.007]&[-0.276,-0.020]\\
\small $99\%$ CI&&[-0.570,0.449]&[-0.374,0.187]&[-0.374,-0.001]&[-0.300,-0.011]\\
			\hline
			\hline
		\end{tabular}
\subcaption{IV: Both=\{Samesex,Twins\}}
\centering
\begin{tabular}{lccccc}
			\hline
			\hline
&&Manski&HV&Chesher&SV\cr
\hline
\small HMUE&&[-0.560,0.439]&[-0.294,0.100]&[-0.294,-0.064]&[-0.189,-0.103]\\
\small $90\%$ CI&&[-0.566,0.445]&[-0.328,0.134]&[-0.328,-0.030]&[-0.252,-0.048]\\
\small $95\%$ CI&&[-0.567,0.446]&[-0.335,0.142]&[-0.336,-0.022]&[-0.263,-0.038]\\
\small $99\%$ CI&&[-0.570,0.449]&[-0.349,0.156]&[-0.350,-0.008]&[-0.283,-0.026]\\
			\hline
			\hline
		\end{tabular}
\end{center}
\vspace{-0.35cm}
\footnotesize
\begin{tablenotes}
\item Note: The first row of panels (a)-(c) reports the weighted average of the HMUE of the four ATE bounds, and the second to fourth rows report the weighted average of the CLR two-sided CI at different significant levels.
\end{tablenotes}
\end{table}

Table \ref{tab:bounds} reports the weighted average of the HMUE and the CLR two-sided confidence interval (CI) at 90\%, 95\% and 99\% significant level of the four bounds of ATE$(X_P)$, with weights given by the estimated kernel density of $X_P$. Panels (a), (b) and (c) display the results using IV \emph{Samesex}, \emph{Twins} and \emph{Both}, respectively. The weighted average of the Manski bounds estimates in all three panels are essentially identical, since the Manski bounds do not depend on IVs. In all panels, the HV bounds make an improvement over the benchmark Manski bounds, with the HV bound width using \emph{Twins} being narrower than that using \emph{Samesex}, and the HV bound width using \emph{Both} being the narrowest.
The Chesher bounds using \emph{Samesex} fail to identify the sign of the ATE$(X_P)$, as it is a union of both negative and positive intervals. When the IV \emph{Twins} or \emph{Both} is used instead, the weighted average of 95\% CI of the Chesher bounds is $[-0.356,-0.007]$ (using \emph{Twins}) or $[-0.336,-0.022]$ (using \emph{Both}), revealing negative effects of having a third child on women's labor force participation.

For the SV bounds, the results using the IV \emph{Twins} or \emph{Both} dramatically outperform those using \emph{Samesex}. The weighted average of 95\% CIs using \emph{Samesex}, \emph{Twins} and \emph{Both} are $[-0.548,-0.023]$, $[-0.276,-0.020]$ and $[-0.263,-0.038]$, respectively. The SV bounds estimates confirm the negative effect of a third child on women's labor force participation.\footnote{The two-stage least square (2SLS) estimates of \citet[][Table 11]{angrist1998children} give an ATE estimate of -0.123 with 95\% CI of $[-0.178,-0.068]$ using IV \emph{Samesex}, and an estimate of -0.087 with 95\% CI of $[-0.120,-0.054]$ using IV \emph{Twins}. As would be expected, the 95\% two-sided CIs of all four bounds cover the 2SLS estimates and their associated 95\% CIs for both IVs.} To summarize the results above, we can see that for ATE bounds in which the IV plays a key role in extracting identifying information, i.e. HV, Chesher and SV bounds, the IV \emph{Both} gives us the narrowest bounds.

\begin{table}[ht!]
\begin{center}
\caption{Decomposition of Identification Gains and Instrument Identification Power}\label{tab:IG_}
\vspace{-0.35cm}
\subcaption{IV: Samesex}
		\begin{tabular}{lccccc}
			\hline
			\hline
$$&$C_1$&$C_2$&$C_3$&$C_4$&$IIP$\cr
\hline
Based on HMUE&0.439&0.034&0.020&0.507&0.473\\
Based on 90\% CI&0.445&0.025&0.020&0.522&0.471\\
Based on 95\% CI&0.447&0.023&0.020&0.525&0.470\\
Based on 99\% CI&0.449&0.020&0.021&0.530&0.470\\
			\hline
			\hline
		\end{tabular}
\subcaption{IV: Twins}
\begin{tabular}{lccccc}
			\hline
			\hline
&$C_1$&$C_2$&$C_3$&$C_4$&$IIP$\cr
\hline
Based on HMUE&0.439&0.312&0.162&0.086&0.751\\
Based on 90\% CI&0.445&0.235&0.100&0.231&0.681\\
Based on 95\% CI&0.447&0.219&0.094&0.256&0.666\\
Based on 99\% CI&0.449&0.196&0.085&0.290&0.645\\
			\hline
			\hline
		\end{tabular}
\subcaption{IV: Both=\{Samesex,Twins\}}
\begin{tabular}{lccccc}
			\hline
			\hline
&$C_1$&$C_2$&$C_3$&$C_4$&$IIP$\cr
\hline
Based on HMUE&0.439&0.330&0.144&0.086&0.769\\
Based on 90\% CI&0.445&0.268&0.094&0.205&0.713\\
Based on 95\% CI&0.447&0.254&0.088&0.225&0.701\\
Based on 99\% CI&0.449&0.228&0.085&0.257&0.677\\
			\hline
			\hline
		\end{tabular}
\end{center}
\vspace{-0.35cm}
\footnotesize
\begin{tablenotes}
\item Note: $C_1$-$C_4$ and $IIP$ are the weighted average of their associated conditional estimates given $X_P$, with the kernel density of $X_P$ as weights. For both panels (a) to (c), $C_1$ to $C_4$ are computed as described in the footnote \ref{footnote}, and the estimates in each row correspond to different significance levels of the CLR estimation.
\end{tablenotes}
\end{table}
\FloatBarrier

\begin{sidewaysfigure}[htp!]
\centering
\caption{Estimated Bounds of ATE$(x)$}\label{fig:AE1}
\vspace{-0.2cm}
\includegraphics[width=1.05\textwidth]{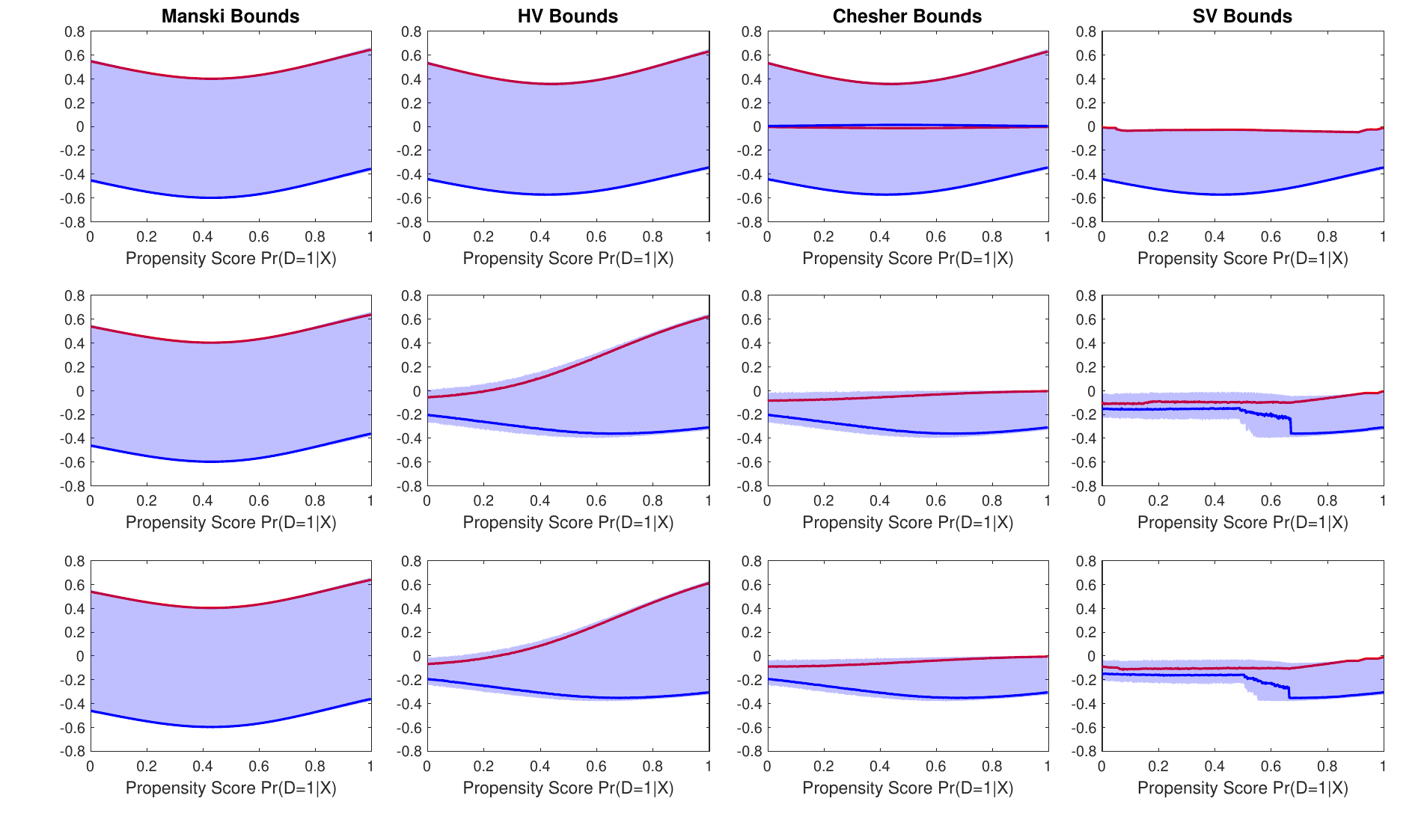}
\vspace{-0.5cm}
\footnotesize
\begin{tablenotes}
\item Note: Panels (a)-(c) plot the estimates ATE$(x)$ as functions of the propensity score $X_P$. The red lines are the upper bounds and blue lines are the lower bounds. The blue shaded area represents the 95\% confidence regions.
\end{tablenotes}
\end{sidewaysfigure}

\begin{sidewaysfigure}[p]
\centering
\caption{Decomposition of Identification Gains}\label{fig:AE2}
\vspace{-0.2cm}
\includegraphics[width=1.05\textwidth]{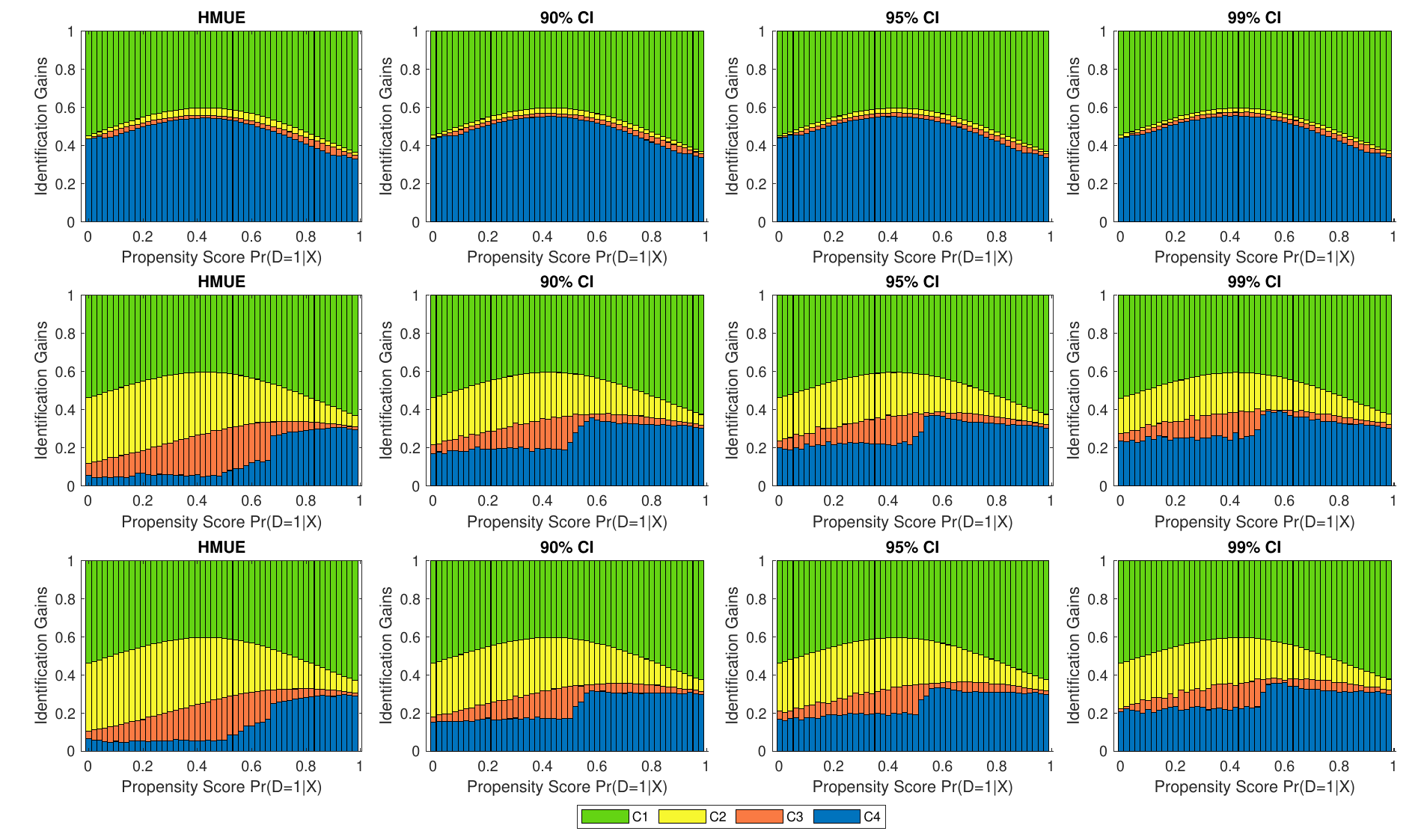}
\vspace{-0.2cm}\footnotesize
\begin{tablenotes}
\item Note: Panels (a)-(c) depict $\hat{C}_j(x)/\hat{\omega}^M(x)$ with $j=1,2,3,4$ against $X_P=\widehat{Pr}(D=1|X)$.
\end{tablenotes}
\end{sidewaysfigure}

The ranking of IV identification power of the three IVs revealed by the discussion above is confirmed by the results in Table \ref{tab:IG_}. The results based on the 95\% CI show that, on average, the identification power of \emph{Twins} (66.6\%) is significantly larger than that of \emph{Samesex} (47.0\%). We can see that whenever \emph{Twins}$=1$ the treatment $D=1$ (i.e. perfect prediction), whereas this is not the case for \emph{Samesex}. It is this feature, of course, that explains the superior performance when the HV, Chesher and SV bounds are evaluated using \emph{Twins} compared to \emph{Samesex}.
When both \emph{Samesex} and \emph{Twins} are used, the identification power of \emph{Both} (70.1\%) exceeds that of either \emph{Samesex} or \emph{Twins}. Closer inspection on the 95\% CIs reveals that any valid IV contributes at least 44.7\% of the SV bounds' improvement relative to the benchmark Manski bounds. Moreover, the identification power of \emph{Twins} is 66.6\%/47.0\%$\approx$1.4 times the identification power of \emph{Samesex} if used individually. In addition, given the contribution of IV validity $C_1$, the extra contribution of \emph{Twins} ($C_2=21.9\%$) via its IV strength to the identification power of \emph{Both} dominates the that of \emph{Samesex} ($C_2=2.3\%$). One remark is that, for the HV and Chesher bounds, IVs with higher $IIP$ clearly lead to narrower bounds.

To explore the heterogeneity of the treatment effects, Figure \ref{fig:AE1} graphs the four bounds of interest against $X_P$.
We can see that when the more powerful of the three IVs are employed, namely \emph{Twins} or \emph{Both}, the HV bounds narrow down the possible range of the ATE$(X_P)$ relative to the benchmark Manski bounds, especially for individuals with a small probability of having a third child. In addition, they can even identify the negative effect for individuals with $X_P$ close to zero. Similar properties are exhibited by the Chesher bounds. The 95\% CI of SV bounds indicate that for women who are less likely to have more than two children, it is more probable that there will be a negative effect on their labor force participation once they have a third child, roughly in the region of -5\% to -20\%. For individuals who are more likely to have more than two children, the effect of having a third child is still negative but with larger possible range, roughly from -5\% to -40\% when $X_P$ is about 0.6, and roughly from 0\% to -30\% when $X_P$ is close to one.

To check the heterogeneity of the IV identification power, Figure \ref{fig:AE2} displays the decompositions plotted against $X_P$. It depicts the fractions of estimated identification gains $\hat C_1(X_P)$ to $\hat C_4(X_P)$ over the estimated Manski bounds width $\hat\omega^M(X_P)$. It is obvious that the IV identification power of \emph{Twins} and \emph{Both} are significantly larger than that of \emph{Samesex}, across all possible values of $X_P$. Furthermore, the contribution of the covariate appears to be amplified when \emph{Twins} is involved in deriving the bounds, leading to a further reduction in the width of the unexplained part relative to the benchmark.

\section{Conclusion}\label{con}
The ATE is a commonly used causal effect measure. Although there have been significant econometric developments, the choice among alternative modelling assumptions and IVs for estimating ATE can be an uninformed one for an empirical researcher without a good understanding of the role of IVs in studying the ATE. This paper aims to illustrate the mechanism through which IVs, and their interactions with treatment endogeneity and exogenous covariates achieve ATE identification gains in a partially identified heterogeneous treatment effect model. Our analysis is carried out by decomposing the identification gains of the ATE bounds in \citet{shaikh2011partial} against the benchmark ATE bounds without IVs in \citet{manski1990nonparametric}. We propose a measure to encapsulate IV identification power, $IIP$, and graphically illustrate the decomposition under alternative scenarios. Our simulation results show that the $IIP$ works well in finite samples as a tool for measuring IV identification power, selecting IVs, and detecting irrelevant ones.

We believe the analysis in this paper brings useful insights to empirical researchers who hope to understand IVs in partially identified models. As an example of these insights, we demonstrate that for the same IV strength, having a high degree of treatment endogeneity that has the opposite sign from that of the ATE can produce greater IV identification power and narrower ATE bounds, while the converse is true if the two have the same sign. In other words, the IV identification power in a partially identified model for binary dependent variables can be very different from the conventional IV strength measure that is only based on the treatment equation. We also discuss practical implications for policy instrument designs after obtaining data from preliminary trials. The analysis in our paper sheds new light on and offers a potential criterion for IV selection in high-dimensional settings within partially identified models. It also raises new questions as to what constitutes an adequate definition of weak IVs and testing for IV validity in conjunction with ATE bounding analyses. Explorations of these issues are left for future research.

\newpage
{\setstretch{1.1}
\bibliography{reference_PID}
}

\newpage
\appendix
\setstretch{1}
\section{Appendix}
Throughout the appendix, $P=P(X,Z)=\text{Pr}[D=1|X,Z]$ with support $\Omega_P$ and $P(x,z)=\text{Pr}[D=1|x,z]$. Let $F$ and $f$ denote the cumulative distribution function and the density function, respectively.
\subsection{Extra Definitions}\label{app_def}
Let $P$ and $P'$ be two independent random variables with the same distribution. For any $x,x'\in\Omega_{X}$, define $H(x,x')=\mathbb{E}[h(x,x',P,P')|P>P']$ where
\begin{align*}
h(x,x',p,p')=&\text{Pr}[Y=1,D=1|x',p]-\text{Pr}[Y=1,D=1|x',p']\\
&~~~~~~~~~~~~-\text{Pr}[Y=1,D=0|x,p']+\text{Pr}[Y=1,D=0|x,p].
\end{align*}
Define $\textbf{X}_{0+}(x)=\{x':H(x,x')\geq0\}$, $\textbf{X}_{0-}(x)=\{x':H(x,x')\leq0\}$, $\textbf{X}_{1+}(x)=\{x':H(x',x)\geq0\}$, and $\textbf{X}_{1-}(x)=\{x':H(x',x)\leq0\}$.

\subsection{The CLR Method}\label{app_CLR}
The CLR half-median-unbiased estimator produces a upper (lower) bound estimator that exceeds (fails below) its true value, each with probability at least a half asymptotically.

Theoretically, the construction of the SV bounds requires the matching of pairs $(x,z)$ and $(x',z')$ such that Pr$[D=1|x,z]=$Pr$[D=1|x',z']$. In practice, it is hard to find such pairs with equal CPS especially when the variation of covariates is limited. In the simulations, the SV bounds are computed by matching $(x,z)$ and $(x',z')$ such that $|\text{Pr}[D=1|x,z]-\text{Pr}[D=1|x',z']|<c$ and $c=1\%$. Although the estimated SV bounds depend on $c$, the estimated $IIP(x)$ does not. Therefore the choice of $c$ has no impacts on the performance of the $IIP(x)$.

\subsection{Lemmas}
\begin{lemma}\label{lemma1}Under Assumption \ref{SVR} (a) and (b), for any $p,p'\in\Omega_{P|x}$ such that $p>p'$, we have
\begin{align*}
&\text{Pr}[D=0|x,p]+\text{Pr}[Y=y,D=1|x,p]-\{\text{Pr}[D=0|x,p']+\text{Pr}[Y=y,D=1|x,p']\}\leq0,\nonumber\\
&\text{Pr}[D=1|x,p]+\text{Pr}[Y=y,D=0|x,p]-\{\text{Pr}[D=1|x,p']+\text{Pr}[Y=y,D=0|x,p']\}\geq0,\nonumber
\end{align*}
for $y\in\{0,1\}$. In addition, \begin{align*}
&\text{Pr}[Y=y,D=1|x,p]-\text{Pr}[Y=y,D=1|x,p']\geq0,\\
&\text{Pr}[Y=y,D=0|x,p]-\text{Pr}[Y=y,D=0|x,p']\leq0,
\end{align*}
for $y\in\{0,1\}$. Lastly, if $\nu_1(1,x)>\nu_1(0,x)$ given $x\in\Omega_X$, then
$\text{Pr}[Y=1|x,p]-\text{Pr}[Y=1|x,p']\geq0$.
If $\nu_1(1,x)\leq\nu_1(0,x)$ given $x\in\Omega_X$, then
$\text{Pr}[Y=1|x,p]-\text{Pr}[Y=1|x,p']\leq0.$
Strict inequalities hold if Assumption \ref{SVR} (e) is imposed on the DGP.
\end{lemma}
\begin{proof}[Proof of Lemma \ref{lemma1}]Under Assumption \ref{SVR} (a) and (b), for $p,p'\in\Omega_{P|x}$ with $p>p'$, we have
\begin{align*}
&\text{Pr}[D=0|x,p]+\text{Pr}[Y=1,D=1|x,p]-\{\text{Pr}[D=0|x,p']+\text{Pr}[Y=1,D=1|x,p']\}\\
=&\text{Pr}[\varepsilon_1<\nu_1(1,x),p'\leq F_{\varepsilon_2}(\varepsilon_2)<p]-\text{Pr}[p'\leq F_{\varepsilon_2}(\varepsilon_2)<p]\\
=&-\text{Pr}[\varepsilon_1\geq\nu_1(1,x),p'\leq F_{\varepsilon_2}(\varepsilon_2)<p]\\
\leq&0.
\end{align*}
Similar manipulations show that
\begin{align*}
&\text{Pr}[D=0|x,p]+\text{Pr}[Y=0,D=1|x,p]-\{\text{Pr}[D=0|x,p']+\text{Pr}[Y=0,D=1|x,p']\}
\leq0,\\
&\text{Pr}[D=1|x,p]+\text{Pr}[Y=1,D=0|x,p]-\{\text{Pr}[D=1|x,p']+\text{Pr}[Y=1,D=0|x,p']\}
\geq0,~\text{and}\\
&\text{Pr}[D=1|x,p]+\text{Pr}[Y=0,D=0|x,p]-\{\text{Pr}[D=1|x,p']+\text{Pr}[Y=0,D=0|x,p']\}
\geq0.
\end{align*}
In addition, using relatively straightforward if somewhat tedious algebra, we can obtain the following inequalities
\begin{align*}
&\text{Pr}[Y=0,D=1|x,p]-\text{Pr}[Y=0,D=1|x,p']
=\text{Pr}[\varepsilon_1\geq\nu_1(1,x),p'\leq F_{\varepsilon_2}(\varepsilon_2)<p]
\geq0,\\
&\text{Pr}[Y=1,D=1|x,p]-\text{Pr}[Y=1,D=1|x,p']
=\text{Pr}[\varepsilon_1<\nu_1(1,x),p'\leq F_{\varepsilon_2}(\varepsilon_2)<p]
\geq0,\\
&\text{Pr}[Y=0,D=0|x,p]-\text{Pr}[Y=0,D=0|x,p']
=-\text{Pr}[\varepsilon_1\geq\nu_1(0,x),p'\leq F_{\varepsilon_2}(\varepsilon_2)<p]
\leq0,~\text{and}\\
&\text{Pr}[Y=1,D=0|x,p]-\text{Pr}[Y=1,D=0|x,p']
=-\text{Pr}[\varepsilon_1<\nu_1(0,x),p'\leq F_{\varepsilon_2}(\varepsilon_2)<p]
\leq0.
\end{align*}
Now suppose that $\nu_1(1,x)>\nu_1(0,x)$ given $x\in\Omega_X$. Then it follows that
\begin{align*}
&\text{Pr}[Y=1|x,p]-\text{Pr}[Y=1|x,p']\nonumber\\
=&\text{Pr}[Y=1,D=1|x,p]+\text{Pr}[Y=1,D=0|x,p]-\text{Pr}[Y=1,D=1|x,p']-\text{Pr}[Y=1,D=0|x,p']\nonumber\\
=&\text{Pr}[\varepsilon_1<\nu_1(1,x),p'\leq F_{\varepsilon_2}(\varepsilon_2)<p]-\text{Pr}[\varepsilon_1<\nu_1(0,x),p'\leq F_{\varepsilon_2}(\varepsilon_2)<p]\nonumber\\
=&\text{Pr}[\nu_1(0,x)\leq \varepsilon_1<\nu_1(1,x),p'\leq F_{\varepsilon_2}(\varepsilon_2)<p]\nonumber\\
\geq&0.
\end{align*}
Finally, using a parallel argument in the case where $\nu_1(1,x)\leq\nu_1(0,x)$ given $x\in\Omega_X$, we can conclude that the inequalities stated in the lemma hold.
\end{proof}

\begin{lemma}\label{lemma2}Under Assumptions \ref{SVR} and \ref{ass2_sv}, joint probabilities Pr$[Y=y,D=d|x,p]$ for $y,d\in\{0,1\}$ are functions of the dependence parameter $\rho$. In addition,
\begin{itemize}
  \item[(a)]$\text{Pr}[Y=1,D=1|x,p]$ and $\text{Pr}[Y=0,D=0|x,p]$ are weakly increasing in $\rho$;
  \item[(b)]$\text{Pr}[Y=1,D=0|x,p]$ and $\text{Pr}[Y=0,D=1|x,p]$ are weakly decreasing in $\rho$.
\end{itemize}
\end{lemma}
\begin{proof}[Proof of Lemma \ref{lemma2}]
For any given $p\in\Omega_P$,
\begin{align}\label{p17}
\text{Pr}[Y=1,D=1|x,p]&=\text{Pr}[\varepsilon_1<\nu_1(1,x),F_{\varepsilon_2}(\varepsilon_2)<p|X=x,P=p]\nonumber\\
&=\text{Pr}[\varepsilon_1<\nu_1(1,x),F_{\varepsilon_2}(\varepsilon_2)<p]\nonumber\\
&=C(F_{\varepsilon_1}(\nu_1(1,x)),p;\rho).
\end{align}
Because the copula $C(\cdot,\cdot;\rho)$ satisfies the concordant ordering with respect to $\rho$, we know that $\text{Pr}[Y=1,D=1|x,p]$ is weakly increasing in $\rho$. Since
\begin{align*}
\text{Pr}[Y=0,D=1|x,p]&=\text{Pr}[D=1|x,p]-\text{Pr}[Y=1,D=1|x,p]=p-C(F_{\varepsilon_1}(\nu_1(1,x)),p;\rho),
\end{align*}
and thus, $\text{Pr}[Y=0,D=1|x,p]$ is weakly decreasing in $\rho$ and $\text{Pr}[Y=1,D=1|x,p]$ is weakly increasing in $\rho$. In addition,
\begin{align}\label{p18}
\text{Pr}[Y=0,D=0|x,p]=&\text{Pr}[\varepsilon_1\geq\nu_1(0,x),F_{\varepsilon_2}(\varepsilon_2)\geq p|x,p]\nonumber\\
=&\text{Pr}[\varepsilon_1\geq\nu_1(0,x),F_{\varepsilon_2}(\varepsilon_2)\geq p]\nonumber\\
=&\text{Pr}[\varepsilon_1\geq\nu_1(0,x)]-\text{Pr}[F_{\varepsilon_2}(\varepsilon_2)< p]+\text{Pr}[\varepsilon_1<\nu_1(0,x),F_{\varepsilon_2}(\varepsilon_2)< p]\nonumber\\
=&1-F_{\varepsilon_1}(\nu_1(0,x))-p+C(F_{\varepsilon_1}(\nu_1(0,x)),p;\rho).
\end{align}
From \eqref{p18} we can see that $\text{Pr}[Y=0,D=0|x,p]$ is weakly increasing in $\rho$, which immediately implies that $\text{Pr}[Y=1,D=0|x,p]$ is weakly decreasing in $\rho$.
\end{proof}

\subsection{Proofs}
\begin{proof}[Proof of Proposition \ref{sv1}] To begin, let us first introduce the following notation:
\begin{equation}\label{four_bounds}
\begin{aligned}
L_0(x,p)&=\text{Pr}[Y=1,D=0|x,p]+\sup_{x'\in\textbf{X}_{0-}(x)}\text{Pr}[Y=1,D=1|x',p],\\
L_1(x,p)&=\text{Pr}[Y=1,D=1|x,p]+\sup_{x'\in \textbf{X}_{1+}(x)}\text{Pr}[Y=1,D=0|x',p],\\
U_0(x,p)&=\text{Pr}[Y=1,D=0|x,p]+p\inf_{x'\in\textbf{X}_{0+}(x)}\text{Pr}[Y=1|x',p,D=1],\\
U_1(x,p)&=\text{Pr}[Y=1,D=1|x,p]+(1-p)\inf_{x'\in\textbf{X}_{1-}(x)}\text{Pr}[Y=1|x',p,D=0].
\end{aligned}
\end{equation}
Then, Theorem 2.1 in \citet[][]{shaikh2011partial} shows that the SV bounds are
\begin{align}
L^{SV}(x)&=L_1(x,\overline{p})-U_0(x,\underline{p})~\text{and}~
U^{SV}(x)=U_1(x,\overline{p})-L_0(x,\underline{p}).
\end{align}

Next we show that $L_0(x,p)$ is weakly decreasing in $p$ (\textit{ceteris paribus}). Under Assumption \ref{SVR} and $\Omega_{X,P}=\Omega_X\times\Omega_P$, for $\forall x\in\Omega_X$ there exists $x^l_0\in\textbf{X}_{0-}(x)$, such that $\nu_1(1,x^l_0)=\sup_{x\in\textbf{X}_{0-}(x)}\nu_1(1,x)$ and
\begin{align*}
L_0(x,p)=\text{Pr}[Y=1,D=0|x,p]+\text{Pr}[Y=1,D=1|x^l_0,p],
\end{align*}
\citep[For detailed particulars see the proof of][Theorem 2.1 (ii)\footnote{The proof is contained in the supplementary material of \citet{shaikh2011partial}.}]{shaikh2011partial}. For $p,p'\in\Omega_P$ and $p'<p$, we have now have
\begin{align}\label{p10}
&L_0(x,p)-L_0(x,p')\nonumber\\
=&\text{Pr}[Y=1,D=0|x,p]+\text{Pr}[Y=1,D=1|x^l_0,p]-\text{Pr}[Y=1,D=0|x,p']-\text{Pr}[Y=1,D=1|x^l_0,p']\nonumber\\
=&\text{Pr}[\varepsilon_1\leq\nu_1(1,x_0^l),p'<\varepsilon_2\leq p)-\text{Pr}[\varepsilon_1\leq\nu_1(0,x),p'<\varepsilon_2\leq p)\nonumber\\
=&\text{Pr}[\nu_1(0,x)<\varepsilon_1\leq\nu_1(1,x_0^l),p'<\varepsilon_2\leq p)\nonumber\\
\leq&0,
\end{align}
where the last inequality follows because $x^l_0\in\textbf{X}_{0-}(x)$, and the Lemma 2.1 in \citet{shaikh2011partial} shows that $x^l_0\in\textbf{X}_{0-}(x)$ implies $\nu_1(1,x^l_0)\leq\nu_1(0,x)$. Thus, from \eqref{p10}, $L_0(x,p)$ is weakly decreasing in $p$.

Similar arguments show that $L_1(x,p)$ is weakly increasing in $p$, $U_0(x,p)$ is weakly increasing in $p$, and $U_1(x,p)$ is weakly decreasing in $p$.
Hence $L^{SV}(x)$ is weakly increasing in $\overline{p}$ and $U^{SV}(x)$ is weakly decreasing in $\overline{p}$. On the other hand, $L^{SV}(x)$ is weakly decreasing in $\underline{p}$ and $U^{SV}(x)$ is weakly increasing in $\underline{p}$. This completes the proof of the proposition.
\end{proof}\\

\begin{proof}[Proof of Proposition \ref{svw}]Consider the case ATE$(x)>0$.
Under Assumption \ref{SVR}, from the definitions of $\textbf{X}_{0+}(x)$, $\textbf{X}_{0-}(x)$, $\textbf{X}_{1+}(x)$ and $\textbf{X}_{1-}(x)$, we know that $\textbf{X}_{0+}(x)$ and $\textbf{X}_{1+}(x)$ are nonempty for $\forall x\in\Omega_X$, since $x$ itself belongs to these two sets. While, $\textbf{X}_{0-}(x)$ and $\textbf{X}_{1-}(x)$ may be empty for some $x\in\Omega_X$. Recall that the supremum and infimum are defined as zero and one over an empty set, respectively. Thus, for the four functions defined in \eqref{four_bounds}, we have
\begin{align}
L_0(x,p)&\geq\text{Pr}[Y=1,D=0|x,p],\nonumber\\
L_1(x,p)&\geq\text{Pr}[Y=1|x,p],\nonumber\\
U_0(x,p)&\leq\text{Pr}[Y=1|x,p],~\text{and}\nonumber\\
U_1(x,p)&\leq\text{Pr}[Y=1,D=1|x,p]+\text{Pr}[D=0|x,p].\nonumber
\end{align}
Therefore, $[L^{SV}(x),U^{SV}(x)]\subset[\underline{L}^{SV}(x),\overline{U}^{SV}(x)]$, where, by Lemma \ref{lemma1}
\begin{equation}\label{svw_p}
\begin{aligned}
\underline{L}^{SV}(x)&=\sup_{p\in\Omega_{P|x}}\text{Pr}[Y=1|x,p]-\inf_{p\in\Omega_{P|x}}\text{Pr}[Y=1|x,p]=\text{Pr}[Y=1|x,\overline{p}(x)]-\text{Pr}[Y=1|x,\underline{p}(x)],\\
\;\\
\overline{U}^{SV}(x)&=\inf_{p\in\Omega_{P|x}}\{\text{Pr}[Y=1,D=1|x,p]+\text{Pr}[D=0|x,p]\}-\sup_{p\in\Omega_{P|x}}\text{Pr}[Y=1,D=0|x,p]\\
&=\text{Pr}[Y=1,D=1|x,\overline{p}(x)]+\text{Pr}[D=0|x,\overline{p}(x)]-\text{Pr}[Y=1,D=0|x,\underline{p}(x)]\end{aligned}
\end{equation}
and it follows that
\begin{align}\label{widest1}
\overline{\omega}(x)=&\text{Pr}[Y=1,D=1|x,\underline{p}(x)]+\text{Pr}[Y=0,D=0|x,\overline{p}(x)].
\end{align}

Now consider the case where ATE$(x)<0$. In contrast to the positive ATE$(x)$ case, $\textbf{X}_{0-}(x)$ and $\textbf{X}_{1-}(x)$ are nonempty for $\forall x\in\Omega_X$ since $x$ itself belongs  to these two sets, while $\textbf{X}_{0+}(x)$ and $\textbf{X}_{1+}(x)$ may be empty for some $x\in\Omega_X$. Thus, the following inequalities hold
\begin{align}
L_0(x,p)&\geq\text{Pr}[Y=1|x,p],\nonumber\\L_1(x,p)&\geq\text{Pr}[Y=1,D=1|x,p],\nonumber\\
U_0(x,p)&\leq\text{Pr}[Y=1,D=0|x,p]+\text{Pr}[D=1|x,p],~\text{and}\nonumber\\U_1(x,p)&\leq\text{Pr}[Y=1|x,p],\nonumber
\end{align}
based on which we can obtain $[L^{SV}(x),U^{SV}(x)]\subset[\underline{L}^{SV}(x),\overline{U}^{SV}(x)]$, where by Lemma \ref{lemma1}
\begin{equation}\label{svw_n}
\begin{aligned}
\overline{U}^{SV}(x)&=\inf_{p\in\Omega_{P|x}}\text{Pr}[Y=1|x,p]-\sup_{p\in\Omega_{P|x}}\text{Pr}[Y=1|x,p]=\text{Pr}[Y=1|x,\overline{p}(x)]-\text{Pr}[Y=1|x,\underline{p}(x)]\\
\;\\
\underline{L}^{SV}(x)&=\sup_{p\in\Omega_{P|x}}\text{Pr}[Y=1,D=1|x,p]-\inf_{p\in\Omega_{P|x}}\{\text{Pr}[Y=1,D=0|x,p]+\text{Pr}[D=1|x,p]\}\\
&=\text{Pr}[Y=1,D=1|x,\overline{p}(x)]-\text{Pr}[Y=1,D=0|x,\underline{p}(x)]-\text{Pr}[D=1|x,\underline{p}(x)]
\end{aligned}
\end{equation}
The width of the ATE outer set is
\begin{align}\label{widest2}
\overline{\omega}(x)=&\text{Pr}[Y=1,D=0|x,\overline{p}(x)]+\text{Pr}[Y=0,D=1|x,\underline{p}(x)].
\end{align}
The desired results follow directly from \eqref{svw_p} to \eqref{widest2} upon application of Lemma \ref{lemma1}.
\end{proof}\newline

\begin{proof}[Proof of Proposition \ref{sv2}]The proof follows directly from the expression for $\overline{\omega}(x)$ in Proposition \ref{svw} and Lemma \ref{lemma2}.
\end{proof}\\

\begin{proof}[Proof of Proposition \ref{prop3_1}]Without loss of generality, assume that $\varepsilon_2$ is uniformly distributed over $[0,1]$. It is sufficient to show that $L^{SV}(x)=U^{SV}(x)=$ATE$(x)$ for any $x\in\mathcal{X}^0\cap\mathcal{X}^1$. First, consider $L^{SV}(x)$. By definition of $\mathcal{X}^1$, there exists a $x^*$ such that $\nu_1(0,x^*)=\nu_1(1,x)$, implying $x^*\in\textbf{X}_{1+}(x)$ and for $p=P(x,z)$
\begin{align*}
\sup_{x'\in \textbf{X}_{1+}(x)}\text{Pr}[Y=1,D=0|x',p]
=&\sup_{x'\in \textbf{X}_{1+}(x)}\text{Pr}[\nu_1(0,x')>\varepsilon_1,p\leq\varepsilon_2|x',p]\nonumber\\
=&\text{Pr}[\nu_1(0,x^*)>\varepsilon_1,p\leq\varepsilon_2|x',p]\nonumber\\
=&\text{Pr}[\nu_1(1,x)>\varepsilon_1,p\leq\varepsilon_2|x,p]\nonumber\\
=&\text{Pr}[Y_1=1,D=0|x,p].
\end{align*}
Similarly, by definition of $\mathcal{X}^0$ we know that $x^*\in\textbf{X}_{0+}(x)$ and for $p=P(x,z)$
\begin{align*}
p\inf_{x'\in\textbf{X}_{0+}(x)}\text{Pr}[Y=1|x',p,D=1]
=&\inf_{x'\in\textbf{X}_{0+}(x)}\text{Pr}[Y=1,D=1|x',p]\nonumber\\
=&\text{Pr}[\nu_1(1,x^*)>\varepsilon_1,p>\varepsilon_2|x',p]\nonumber\\
=&\text{Pr}[\nu_1(0,x)>\varepsilon_1,p>\varepsilon_2|x,p]\nonumber\\
=&\text{Pr}[Y_0=1,D=1|x,p].
\end{align*}
Then, according to the expression of $L^{SV}(x)$,

\begin{align*}
L^{SV}(x)=&\sup_{p\in\Omega_{P|x}}\left\{\text{Pr}[Y_1=1,D=1|x,p]+\text{Pr}[Y_1=1,D=0|x',p]\right\}\\
&~~~~-\inf_{p\in\Omega_{P|x}}\left\{\text{Pr}[Y_0=1,D=0|x,p]+\text{Pr}[Y_0=1,D=1|x',p]\right\}\\
=&\sup_{p\in\Omega_{P|x}}\left\{\text{Pr}[Y_1=1|x,p]\right\}-\inf_{p\in\Omega_{P|x}}\left\{\text{Pr}[Y_0=1|x,p]\right\}\\
=&\text{Pr}[Y_1=1|x]-\text{Pr}[Y_0=1|x],
\end{align*}
where the last equality comes from the independence of $Z$ to $(Y_1,Y_0)$. Parallel arguments can be applied to show that $U^{SV}(x)=$ATE$(x)$.
\end{proof}\\

\begin{proof}[Proof of Proposition \ref{prop3_2}](i)The proof follows directly from the expression of the SV bounds.

(ii) Degeneracy of $\nu_1(D,X)|D$ indicates that there exists a function $m_1:\{0,1\}\mapsto\mathbb{R}$ such that $\nu_1(d,x)=m_1(d)$ for all $(d,x)\in\{0,1\}\times\Omega_X$. Take ATE$(x)$ to be positive. When $H(x,x')$ is well defined and $\nu_1(D,X)=m_1(D)$, $\textbf{X}_{0+}(x)=\textbf{X}_{1+}(x)=\Omega_X$, and $\textbf{X}_{0-}(x)=\textbf{X}_{1-}(x)=\emptyset$. Since $\varepsilon_2$ is continuously distributed, we can conclude that $\text{Pr}[D=1|x',z']=\text{Pr}[D=1|x,z]$ implies $\nu_2(x,z)=\nu_2(x',z')$.

For $L^{SV}(x)$, first, consider $\sup_{x'\in \textbf{X}_{1+}(x)}\text{Pr}[Y=1,D=0|x',p]$.
Since $\textbf{X}_{1+}(x)$ equals $\Omega_X$ because $\nu_1(D,X)=m_1(D)$, we have $\text{Pr}[D=1|x',z']=p$ for at least $(x',z')=(x,z)$, and thus $\sup_{x'\in \textbf{X}_{1+}(x)}\text{Pr}[Y=1,D=0|x',p]$ is well-defined. It follows that
\begin{align}\label{p37}
\sup_{x'\in \textbf{X}_{1+}(x)}\text{Pr}[Y=1,D=0|x',p]
=&\sup_{x'\in \textbf{X}_{1+}(x)}\text{Pr}[\nu_1(0,x')>\varepsilon_1,\nu_2(x',z')\leq\varepsilon_2|x',p]\nonumber\\
=&\sup_{x'\in \textbf{X}_{1+}(x)}\text{Pr}[m_1(0)>\varepsilon_1,\nu_2(x,z)\leq\varepsilon_2|x',p]\nonumber\\
=&\sup_{x'\in \textbf{X}_{1+}(x)}\text{Pr}[m_1(0)>\varepsilon_1,\nu_2(x,z)\leq\varepsilon_2|x,p]\nonumber\\
=&\text{Pr}[Y=1,D=0|x,p],
\end{align}
where the second equality arises because $\nu_1(0,x')=m_1(0)$ and the third equality is due to the independence of $(X,Z)$ to $(\varepsilon_1,\varepsilon_2)$. Similarly, we can show that
\begin{align}\label{p38}
p\inf_{x'\in\textbf{X}_{0+}(x)}\text{Pr}[Y=1|x',p,D=1]
=&\text{Pr}[Y=1,D=1|x,p].
\end{align}
By virtue of equations \eqref{p37} and \eqref{p38}, and Lemma \ref{lemma1}, $L^{SV}(x)$ can be rewritten as
\begin{align}\label{p39}
L^{SV}(x)&=\sup_{p\in\Omega_{P|x}}\left\{\text{Pr}[Y=1,D=1|x,p]+\text{Pr}[Y=1,D=0|x,p]\right\}\nonumber\\
&~~~~~~~~~~~~-\inf_{p\in\Omega_{P|x}}\left\{\text{Pr}[Y=1,D=0|x,p]+\text{Pr}[Y=1,D=1|x,p]\right\}\nonumber\\
&=\sup_{p\in\Omega_{P|x}}\text{Pr}[Y=1|x,p]-\inf_{p\in\Omega_{P|x}}\text{Pr}[Y=1|x,p]\nonumber\\
&=\text{Pr}[Y=1|x,\overline{p}(x)]-\text{Pr}[Y=1|x,\underline{p}(x)].
\end{align}
For $U^{SV}(x)$, because $\textbf{X}_{0-}(x)$ and $\textbf{X}_{1-}(x)$ are empty, from Lemma \ref{lemma1} we get
\begin{align}\label{p43}
U^{SV}(x)&=\inf_{p\in\Omega_{P|x}}\left\{\text{Pr}[Y=1,D=1|x,p]+(1-p)\right\}-\sup_{p\in\Omega_{P|x}}\text{Pr}[Y=1,D=0|x,p]\nonumber\\
&=\text{Pr}[Y=1,D=1|x,\overline{p}(x)]+(1-\overline{p}(x))-\text{Pr}[Y=1,D=0|x,\underline{p}(x)].
\end{align}
It yields from \eqref{p39} and \eqref{p43} that
\begin{align*}
\omega^{SV}
=&\text{Pr}[Y=0,D=0|x,\overline{p}(x)]+\text{Pr}[Y=1,D=1|x,\underline{p}(x)],
\end{align*}
which is equal to $\overline{\omega}(x)$. The proof for the negative ATE$(x)$ case is completely analogous.
\end{proof}\\

\begin{proof}[Proof of Proposition \ref{coro2}](i) If $Z$ is irrelevant, by definition $IIP(x)=0$ and the SV bounds reduce to the benchmark Manski bounds \citep[Remark 2.1 in][]{shaikh2011partial}. To establish necessity we will show that the events $Z$ is relevant and $IIP(x)=0$ occur simultaneously leads to a contradiction. If $Z$ is relevant, then $IIP(x)=1-\overline{\omega}(x)$. The goal, therefore, is to show that relevant $Z$ leads to $\overline{\omega}(x)<1$ if ATE$(x)>0$. By Lemma \ref{lemma1}
\begin{align}\label{ineq1}
&\text{Pr}[Y=0,D=0|x,\overline{p}(x)]-\text{Pr}[Y=0,D=0|x,\underline{p}(x)]
=-\text{Pr}\left[\varepsilon_1\geq\mu_1(0,x),\underline{p}(x)\leq \varepsilon_2<\overline{p}(x)\right]<0,
\end{align}
where the relevance of $Z$ guarantees that $\underline{p}(x)<\overline{p}(x)$ and the continuity of the joint distribution of $(\varepsilon_1,\varepsilon_2)$ implies that \eqref{ineq1} is strictly negative. The result for the case ATE$(x)<0$ can be verified analogously. Therefore,
$\overline{\omega}(x)<1,$ leading to $IIP(x)>0$.\newline

\noindent(ii) If $Z$ is a perfect predictor of $D$, then there exist a $p^*$ and a $p^{**}$ such that $\text{Pr}(D=1|x,p^*)=0$ and $\text{Pr}(D=1|x,p^{**})=1$, which obviously implies $Z$ is relevant and $IIP(x)=1-\overline{\omega}(x)$. Let $\underline{p}(x)=p^*$ and $\overline{p}(x)=p^{**}$. From the expression for $\overline{\omega}(x)$, perfect prediction of $Z$ leads to $\overline{\omega}(x)=0$ and $IIP(x)=1$ for both ATE$(x)>0$ and ATE$(x)<0$.

Moreover, since $0\leq\omega^{SV}(x)\leq\overline{\omega}(x)$, $\overline{\omega}(x)=0$ implies $\omega^{SV}(x)=0$. The ATE$(x)$ is point identified if $IIP(x)=1$.
\end{proof}

\end{document}